\pdfoutput=1
\documentclass[11pts]{article}

\usepackage{enumerate}
\usepackage{hyperref}
\usepackage[utf8]{inputenc}
\usepackage[T1]{fontenc}
\usepackage[english]{babel}
\usepackage{url}
\usepackage[squaren, Gray, cdot]{SIunits}
\usepackage{array}
\usepackage{caption}
\usepackage{layout}
\usepackage{csquotes}
\usepackage[top=3.2cm,bottom=3.2cm,left=3.2cm,right=3.2cm]{geometry}

\usepackage{amsthm}
\usepackage{amsmath}
\usepackage{amssymb}
\usepackage{thm-restate}

\newtheorem{theorem}{Theorem}
\newtheorem{definition}[theorem]{Definition}
\newtheorem{lemma}[theorem]{Lemma}
\newtheorem{proposition}[theorem]{Proposition}
\newtheorem{claim}[theorem]{Claim}
\newtheorem{remark}[theorem]{Remark}

\usepackage{listings}

\usepackage{graphicx}
\usepackage{subfigure}

\usepackage{bbm}
\usepackage{stmaryrd}
\usepackage{tikz}

\usepackage{algorithm}
\usepackage[noend]{algpseudocode}

\newcommand{\wdeg}{\mathrm{\mathbf{w}\!\deg}}

\allowdisplaybreaks

\title{Maximum Weight $b$-Matchings in Random-Order Streams}

\date{}

\author{Chien-Chung Huang\\
\textit{CNRS, DI ENS, École normale supérieure, Université PSL, France}
\and François Sellier\\
\textit{Université Paris Cité, CNRS, IRIF, Paris, France}\\
\textit{Mines Paris, Université PSL, Paris, France}}

\begin{document}

\maketitle

\begin{abstract}
    We consider the maximum weight $b$-matching problem in the random-order semi-streaming model. Assuming all weights are small integers drawn from $[1,W]$, we present a $2 - \frac{1}{2W} + \varepsilon$ approximation algorithm working on randomly-ordered streams of edges with high probability and using $O(\max(|M_G|, n) \cdot poly(\log(m),W,1/\varepsilon))$ memory, where $|M_G|$ denotes the cardinality of the optimal matching. Our result generalizes that of Bernstein~\cite{bernstein:LIPIcs:2020:12419}, which achieves a $3/2 + \varepsilon$ approximation for the maximum cardinality simple matching. When $W$ is small, our result also improves upon that of Gamlath \emph{et al.}~\cite{GamlathSSS2019}, which obtains a $2 - \delta$ approximation (for some small constant $\delta \sim 10^{-17}$) for the maximum weight simple matching. In particular, for the weighted $b$-matching problem, ours is the first result beating the approximation ratio of $2$. Our technique hinges on a generalized weighted version of edge-degree constrained subgraphs, originally developed by Bernstein and Stein~\cite{BernsteinS15}. Such a subgraph has bounded vertex degree (hence uses only a small number of edges), and can be easily computed. The fact that it contains a $2 - \frac{1}{2W} + \varepsilon$ approximation of the maximum weight matching is proved using the classical K\H{o}nig-Egerváry's duality theorem. 
\end{abstract}

\paragraph*{Keywords} Maximum weight matching, $b$-matching, streaming, random order

\paragraph*{Funding} This work was funded by the grants ANR-19-CE48-0016 and ANR-18-CE40-0025-01 from the French National Research Agency (ANR).

\paragraph*{Related version} Part of this work appeared at ESA 2022~\cite{Huang_et_al:LIPIcs.ESA.2022.68}.

\paragraph*{Acknowledgements} The authors thank Claire Mathieu and the anonymous reviewers for their helpful comments.

\section{Introduction}

    The maximum weight ($b$-)matching problem is a classical problem in combinatorial optimization. In this paper we will study a sparsifier for that problem and use it in order to design a streaming algorithm for randomly-ordered streams of edges.
    
    \paragraph*{Structural Result for Edge-Degree Constrained Subgraphs}
    
    Our main tool is a generalized weighted version of the \emph{edge-degree constrained subgraph} (EDCS), a graph sparsifier originally designed for the maximum matching problem by Bernstein and Stein~\cite{BernsteinS15}. Let us first recall the definition an EDCS $H$ of a graph $G$~\cite{BernsteinS15}.
    
    \begin{definition}[from \cite{BernsteinS15}]
        \label{def:intro-edcs}
         Let $G = (V,E)$ be a graph, and $H$ a subgraph of $G$. Given any integer parameters $\beta \geq 2$ and $\beta^- \leq \beta - 1$, we say that $H$ is a $(\beta, \beta^-)$-EDCS of $G$ if $H$ satisfies the following properties:
        \begin{enumerate}[(i)]
            \item \makebox[12em][l]{For any edge $(u, v) \in H$,} $\deg_H(u) + \deg_H(v) \leq \beta$
            \item \makebox[12em][l]{For any edge $(u, v) \in G \backslash H$,} $\deg_H(u) + \deg_H(v) \geq \beta^-$.
        \end{enumerate}
    \end{definition}
    
    An EDCS has a size that can be easily controlled by the parameter $\beta$ and it somehow ``balances'' the vertex degrees in the graph. A very nice property of this sparsifier is that, for well-chosen values of $\beta$ and $\beta^-$, it always contains a $3/2 + \varepsilon$ approximation of the maximum cardinality matching~\cite{AssadiB19,BernsteinS16}:
    
    \begin{theorem}[from the recent work of Assadi and Bernstein~\cite{AssadiB19}] \label{thm:intro-general-edcs}
        Let $0 < \varepsilon < 1/2$. Set $\lambda = \frac{\varepsilon}{32}$. Let $\beta \geq \beta^- + 1$ be integers such that $\beta \geq 8\lambda^{-2}\log(1/\lambda)$ and $\beta^- \geq (1 - \lambda) \cdot \beta$. Then any $(\beta, \beta^-)$-EDCS $H$ of a graph $G$ contains a matching $M_H$ such that $\left(\frac{3}{2} + \varepsilon\right)\cdot |M_H| \geq |M_G|$ where $M_G$ denotes a maximum cardinality matching.
    \end{theorem}
    
    In our paper, we generalize the EDCS in two ways:
    \begin{itemize}
        \item we handle (small) integer-weighted edges;
        \item we handle the more general case of $b$-matchings.
    \end{itemize}

    To describe our generalization, first let us introduce some notation. A weighted multi-graph $G = (V, E)$ is defined by its set of vertices $V$ and its multi-set of weighted edges $E$ drawn from $V \times V \times \{1,2,\dots,W\}$ (\emph{i.e.}, $e = (u,v,k)$ represents an edge between $u$ and $v$ of weight $w(e) = k$; edges are non-oriented). We emphasize that $E$ is a multi-set: not only can there be multiple edges between two vertices, but also some of these edges can have the same weight. We assume that the multi-graph does not contain any self-loop. For a given vertex $v \in V$ and a given subgraph $H$ of $G$, $\delta_H(v)$ denotes the multi-set of edges incident to $v$ in $H$, $\deg_H(v)$ the \emph{degree} of $v$ in the multi-graph $H$ and $\wdeg_H(v)$ its \emph{weighted degree} $\sum_{(u,v,w)\in \delta_H(v)}w$ in $H$. Given a weighted multi-graph $G = (V,E)$ and a set of capacities $\{b_v\}_{v \in V} \in \mathbb{Z}_{+}^V$ associated to vertices $v \in V$, a multi-set of weighted edges $M$ is called a \emph{$b$-matching} if for all $v \in V$ the number of edges incident to $v$ in $M$ is smaller than or equal to $b_v$ (a set of edges is called a \emph{matching} when all the $b_v$s are equal to $1$). For a given subgraph $H$ of $G$, we denote by $M_H$ an arbitrary maximum weight $b$-matching included in $H$. The concept of $b$-matching encompasses that of matching and allows us to tackle a larger variety of real situations where the vertices have different capacities, \emph{e.g.}~\cite{tennenholtz2000some}. In this paper we will assume that the number of edges between any two vertices $u$ and $v$ is at most $\min(b_u, b_v)$.\footnote{This is actually a reasonable assumption as the maximum number of edges that are relevant between two given vertices $u$ and $v$ to construct a $b$-matching is at most $\min(b_u, b_v)$. This assumption is more debatable in the streaming setting, and this is why we explain how to handle this case in Appendix~\ref{app:relevant-edges}.}
    
    \begin{definition}
        \label{def:intro-w-b-edcs}
        Let $G = (V,E)$ be a weighted multi-graph, $\{b_v\}_{v \in V}$ be a set of constraints, and $H$ be a subgraph of $G$. Given any integer parameters $\beta \geq 3$ and $\beta^- \leq \beta - 2$, we say that $H$ is a $(\beta, \beta^-)$-$w$-$b$-EDCS of $G$ if $H$ satisfies the following properties:
        \begin{enumerate}[(i)]
            \item \makebox[14em][l]{For any edge $(u, v, w_{uv}) \in H$,} $\frac{\wdeg_H(u)}{b_u} + \frac{\wdeg_H(v)}{b_v} \leq \beta \cdot w_{uv}$
            \item \makebox[14em][l]{For any edge $(u, v, w_{uv}) \in G \backslash H$,} $\frac{\wdeg_H(u)}{b_u} + \frac{\wdeg_H(v)}{b_v} \geq \beta^- \cdot w_{uv}$. 
        \end{enumerate}
    \end{definition}
    
    An EDCS is a special case of a \emph{weighted $b$-EDCS} when all the $b_v$s and all the weights are equal to $1$. 
    We can show that such $w$-$b$-EDCSes as described in Definition~\ref{def:intro-w-b-edcs} always exist (Propositions~\ref{prop:exist-w-edcs} and~\ref{prop:exist-w-b-edcs}). Moreover, we can also prove that it uses only a reasonable number of edges (up to $2\beta \cdot |M_G|$) and that a relatively large weighted $b$-matching can be found in it:
    
    \begin{restatable}[]{theorem}{generalwbedcs} \label{thm:intro-general-w-b-edcs}
        Let $0 < \varepsilon < 1/2$ and let $W$ be an integer parameter. Set $\lambda = \frac{\varepsilon}{100W}$. Let $\beta \geq \beta^- + 2$ be integers such that $\frac{\beta + 6W}{\log(\beta + 6W)} \geq 2 W^2\lambda^{-2}$ and $\beta^- - 6W \geq (1 - \lambda) \cdot (\beta + 6W)$. Then any $(\beta, \beta^-)$-$w$-$b$-EDCS $H$ of an edge-weighted multi-graph $G$ with integer edge weights bounded by $W$ contains a $b$-matching $M_H$ such that $\left(2 - \frac{1}{2W} + \varepsilon\right)\cdot w(M_H) \geq w(M_G)$.
    \end{restatable}
    
    Moreover, in Appendix~\ref{app:tight-bound-edcs}, we give a whole class of tight examples reaching the bound of Theorem~\ref{thm:intro-general-w-b-edcs}. Comparing with the previous result of~\cite{AssadiB19, BernsteinS16} (Theorem~\ref{thm:intro-general-edcs}), setting $W = 1$, we can observe that the approximation ratio is the same, but the constraints on $\beta$ and $\beta^-$ are a bit stricter here. However, in our case we can deal with $b$-matchings as well, and we can handle situations in which $W > 1$. 
    
    \paragraph*{Algorithmic Result in Random-Order Streams}
    
    The \emph{semi-streaming} model of computation~\cite{FKMSZ2005} has been motivated by the recent rise of massive graphs, for which we cannot afford to store all the edges in memory while it is still achievable to store an output of the considered problem in memory. Given that the graph is made of $|V| = n$ vertices and $|E| = m$ edges, in the semi-streaming model the graph is presented to the algorithm as a stream of edges $e_1, \ldots, e_m$. The algorithm is allowed to make a single pass over that stream and can use a memory roughly proportional to the output size, or proportional to the number of vertices in the graph, up to a poly-logarithmic factor. 
    
    We note that in the most general model where an adversary decides the order of the elements, even for the maximum cardinality simple matching, it is still unclear whether it is possible to beat the approximation ratio of $2$. 
    
    Our focus here is on the \emph{random-order} streaming model, where the permutation of the edges in the stream is assumed to be chosen uniformly at random. This is a quite reasonable assumption as real-world data have little reason of being ordered in an adversarial way, and a uniform random order is a simple and convenient approximation of this. In particular, for this model, it is possible to beat the approximation factor of $2$~\cite{KonradMM12,GamlathSSS2019}, at least for the simple matching. 
    
    Using an adaptation of EDCS, Bernstein~\cite{bernstein:LIPIcs:2020:12419} obtained a $3/2+\varepsilon$ approximation in the random-order semi-streaming framework (with probability $1 - 2n^{-3}$ and using $O(n\cdot \log(n) \cdot poly(1/\varepsilon))$ memory). Similarly, we can adapt our $w$-$b$-EDCSes to design a semi-streaming algorithm for randomly-ordered streams of weighted edges:
    
    \begin{theorem} \label{thm:intro-streaming}
        Let $0 < \varepsilon < \frac{1}{2}$ and let $W$ be an integer parameter. There exists an algorithm that can extract with high probability (at least $1-2m^{-3}$) from a randomly-ordered stream of weighted edges having integer weights in $\{1, \ldots, W\}$ a weighted $b$-matching with an approximation ratio of $2 - \frac{1}{2W} + \varepsilon$, using $O(\max(|M_G|,n) \cdot poly(\log(m), W, 1/\varepsilon))$ memory.
    \end{theorem}
    
    Theorem~\ref{thm:intro-streaming} is the first result for the maximum (integer-weighted) $b$-matching problem in the random-order semi-streaming framework. For the special case of simple matching, when $W=1$, we essentially re-capture the result of  Bernstein~\cite{bernstein:LIPIcs:2020:12419} (albeit using slightly more memory). When $W>1$, we note that prior to our work, Gamlath \emph{et al.}~\cite{GamlathSSS2019} have obtained an approximation ratio of $2 - \delta$ for some small $\delta \sim 10^{-17}$. Our result gives a better approximation when $W$ is reasonably small (but using a memory depending polynomially in $W$) and we believe that our approach is significantly simpler. 
    
    \begin{remark}
        Another generalization of EDCS has been developed by Bernstein \emph{et al.}~\cite{BernsteinDL21} to maintain a $3/2 + \varepsilon$ approximation of the optimal weighted matching in a dynamic graph. However it is still unknown if their construction can actually lead to an algorithm in the random-order one-pass semi-streaming model~\cite{BernsteinDL21}, or applied to $b$-matchings.
    \end{remark}

    \paragraph*{Technical Overview}
    
        To generalize the EDCS to the weighted case, a very natural first idea is to build multiple EDCSes, one for each edge-weight from $1$ to $W$, and then take their union. We show in Appendix~\ref{app:copies-edcs} that such an idea does not lead to a subgraph containing a matching that is better than a $2$ approximation. 
    
        Our approach is a proper generalization of EDCS, as defined in Definition~\ref{def:w-b-edcs}. 
        In Theorem~\ref{thm:intro-general-w-b-edcs}, we show that such a $w$-$b$-EDCS contains a matching of good approximation ratio. The proof of this theorem is technically the most innovative part of the present work. In order to handle integer-weighted matchings (see Section~\ref{sec:weighted-matching}) we use K\H{o}nig-Egerváry's duality theorem~\cite{Egervary1931matrixok} and the construction of an auxiliary graph. The fact that the weights of the edges are integers is critical to get an approximation ratio better than $2$ (especially for Claim~\ref{claim:average}). Then, to handle $b$-matchings (see Section~\ref{sec:weighted-b-matching}), we build a reduction to simple matchings and show that by doing so we do not lose too much in the approximation ratio.
        
        Regarding Theorem~\ref{thm:intro-streaming}, when we design a semi-streaming algorithm to extract a $b$-matching there is an additional challenge: we do not know in advance the actual size of $M_G$, which cannot be bounded by $n$ (for instance $|M_G|$ could be of size $n^{1.2}$ or even larger), but we still want to use as little memory as possible, \emph{i.e.} $O(\max(|M_G|, n) \cdot poly(\log(m),W,1/\varepsilon))$. We tackle this issue using a guessing technique in the early phase of the stream (see Section~\ref{sec:random-order}).

    \paragraph*{Related Work}
    
    In the adversarial semi-streaming setting, for the unweighted case, the simple greedy algorithm building a maximal matching provides a $2$ approximation, which is the best known approximation ratio. Knowing whether it is possible to achieve a better approximation ratio is a major open question in the field of streaming algorithms. For weighted matchings an approximation ratio of $2 + \varepsilon$ can be achieved~\cite{GW2019,LW2021,PS2017}. For weighted $b$-matchings the approximation ratio $2 + \varepsilon$ can also be attained~\cite{HuangS21}. On the hardness side, we know that an approximation ratio better than $1 + \ln 2 \approx 1.69$ cannot be achieved~\cite{Kapralov2021}.
        
    In contrast, for the \emph{random-order} stream, a first result was obtained by Konrad, Magniez, and Mathieu~\cite{KonradMM12} with an approximation ratio strictly below $2$ for unweighted simple matchings. The approximation ratio was then improved in a sequence of papers~\cite{GamlathSSS2019,Konrad18,FarhadiHMRR20,bernstein:LIPIcs:2020:12419}. Currently the best result is due to Assadi and Behnezhad~\cite{AssadiB21}, who obtained the ratio of $3/2 - \delta$ for some small constant $\delta \sim 10^{-14}$. Regarding weighted simple matchings, Gamlath \emph{et al.}~\cite{GamlathSSS2019} obtained an approximation ratio of $2 - \delta$ for some small constant $\delta \sim 10^{-17}$. Regarding $b$-matchings, to our knowledge the only result is an approximation ratio of $2 - \delta$ \emph{in expectation} for random-order online matroid intersection by Guruganesh and Singla~\cite{GuruganeshS17} (hence it applies for unweighted bipartite $b$-matchings).

\section{EDCS for Weighted Matchings}
    \label{sec:weighted-matching}

    In this section we consider the problem of finding a maximum weight matching in an edge-weighted graph $G = (V, E)$ where the edges have integer weights in $[1,W]$. For ease of presentation, we will use simplified notations for \emph{simple} graphs in this section. Here $w(u,v)$ denotes edge weight between vertices $u$ and $v$. For a subgraph $H$ of $G$ and a vertex $u \in V$, we denote by $N_H(v)$ the set of vertices adjacent to $v$ in $H$, by $\deg_H(v)$ the degree of $v$ in $H$, \emph{i.e.}, $\deg_H(v) = |N_H(v)|$, and by $\wdeg_H(v)$ the weighted degree of $v$ in $H$, \emph{i.e.}, $\wdeg_H(v) = \sum_{u\in N_H(v)} w(u,v)$. For a subgraph $H$ of $G$, we will denote by $M_H$ an arbitrary maximum weight matching in $H$. Then we define the notion of edge-degree constrained subgraphs for weighted graphs ($w$-EDCS), which in fact is just Definition~\ref{def:intro-w-b-edcs} specialized to the setting in this section. 
    \begin{definition}
        \label{def:w-edcs}
         Let $G = (V,E)$ be a graph with weighted edges, and $H$ be a subgraph of $G$. Given any integer parameters $\beta \geq 3$ and $\beta^- \leq \beta - 2$, we say that $H$ is a $(\beta, \beta^-)$-$w$-EDCS of $G$ if $H$ satisfies the following properties:
        \begin{enumerate}[(i)]
            \item \makebox[12em][l]{For any edge $(u, v) \in H$,} $\wdeg_H(u) + \wdeg_H(v) \leq \beta \cdot w(u,v)$
            \item \makebox[12em][l]{For any edge $(u, v) \in G \backslash H$,} $\wdeg_H(u) + \wdeg_H(v) \geq \beta^- \cdot w(u,v)$.
        \end{enumerate}
    \end{definition}
    
    Here is a first simple proposition on $(\beta,\beta^-)$-$w$-EDCS (coming from Property~(i)).
    \begin{proposition}
        \label{prop:w-edcs-bounded-degree}
        Let $H$ be a $(\beta,\beta^-)$-$w$-EDCS of a given graph $G$. Then, for all $v \in V$, we have $\deg_H(v) \leq \beta$.
    \end{proposition}
    
    \begin{proof}
        Let $v \in V$. If $N_H(v) = \emptyset$, the stated property is trivial. Otherwise, pick a vertex $u$ such that $w(u,v) = \min_{u' \in N_H(v)}w(u', v)$. Then, by Property~(i), $\beta \cdot w(u, v) \geq \wdeg_H(v)$. Therefore, $\deg_H(v) \leq \frac{\wdeg_H(v)}{w(u,v)} \leq \beta$, because we have chosen $u$ so that any edge incident to $v$ in $H$ has a weight larger or equal to $w(u,v)$. 
    \end{proof}
    
    We show the existence of $w$-EDCSes by construction, using a local search 
    algorithm. The following proof closely follows the argument of~\cite{AssadiB19}.
    
    \begin{proposition} \label{prop:exist-w-edcs}
        Any graph $G = (V, E)$ with weighted edges contains a $(\beta, \beta^-)$-$w$-EDCS for any parameters $\beta \geq \beta^- + 2$. Such a $(\beta, \beta^-)$-$w$-EDCS can be found in $O(\beta^2 W^2 \cdot n)$ time.
    \end{proposition}
    
    \begin{proof}
        Start with an empty subgraph $H$. Then try the following local improvements of $H$, until it is no longer possible. If there is an edge in $H$ violating Property~(i) of Definition~\ref{def:w-edcs}, then fix that edge by removing it from $H$. Otherwise, if there is an edge in $G \backslash H$ violating Property~(ii), then fix that edge by inserting it in $H$.
        
        Observe that we give the priority to the correction of violations of Property~(i), so that at each step of the algorithm all the vertices have degrees bounded by $\beta + 1$ (as after inserting an edge, Proposition~\ref{prop:w-edcs-bounded-degree} may be violated).
        To prove that this algorithm terminates in finite time and to show the existence of a $w$-EDCS, we introduce a potential function:
        \[\Phi(H) = (2\beta - 2) \sum_{(u,v) \in H}w(u,v)^2 - \sum_{u \in V} (\wdeg_H(u))^2.\]
        As the vertices have degrees bounded by $\beta + 1$ and the edges have weights bounded by $W$, the value of that potential function is bounded by $2\beta^2 W^2 \cdot n$. Then we can show that after each local improvement step, the value of $\Phi(H)$ increases at least by $2$.
        
        In fact, when Property~(i) is not satisfied by some edge $(u, v) \in H$, we have $\wdeg_H(u) + \wdeg_H(v) \geq \beta \cdot w(u,v) + 1$. If we fix it by erasing this edge from $H$ the value of $\Phi(H)$ is increased by $-(2\beta - 2) \cdot w(u,v)^2 + 2 \cdot \wdeg_H(u) \cdot w(u,v)  + 2 \cdot \wdeg_H(v) \cdot w(u,v) - 2 \cdot w(u,v)^2 \geq -(2\beta - 2) \cdot w(u,v)^2 + 2 \cdot (\beta \cdot w(u,v) + 1) \cdot w(u,v) - 2 \cdot w(u,v)^2 = 2 \cdot w(u,v) \geq 2$. 
            
        When Property~(ii) is not satisfied by some edge $(u, v) \in G \backslash H$, we have $\wdeg_H(u) + \wdeg_H(v) \leq \beta^- \cdot w(u,v) - 1$. If we insert $(u, v)$ in $H$, the value of $\Phi(H)$ increases by $(2\beta - 2) \cdot w(u,v)^2 - 2\cdot \wdeg_H(u) \cdot w(u,v) - 2 \cdot \wdeg_H(v) \cdot w(u,v) - 2 \cdot w(u,v)^2 \geq (2\beta - 2) \cdot w(u,v)^2 - 2\cdot( \beta^- \cdot w(u,v) - 1) \cdot w(u,v) - 2\cdot w(u,v)^2 \geq (2\cdot (\beta - \beta^-) - 4) \cdot w(u,v)^2 + 2 \cdot w(u,v)\geq 2 \cdot w(u,v) \geq 2$, as $\beta \geq \beta^- + 2$.
        
        Therefore, the algorithm terminates in $O(\beta^2 W^2\cdot n)$ steps.
    \end{proof}   
        
    We also introduce the notion of $w$-vertex-cover of the edge-weighted graph.
    \begin{definition}
        We say that the non-negative integer variables $(\alpha_v)_{v \in V}$ represent a $w$-vertex-cover of a subgraph $H$ of $G$ if for all $(u,v) \in H$, we have $w(u,v) \leq \alpha_u + \alpha_v$. The sum $\sum_{v \in V}\alpha_v$ is called the weight of the $w$-vertex-cover.
    \end{definition}
    
    To prove the apprixmation guarantee for the maximum weight problem, we will use the theorem of K\H{o}nig-Egerváry~\cite{Egervary1931matrixok}, which is a classic duality theorem.
    \begin{theorem}[K\H{o}nig-Egerváry] \label{thm:konig-egervary}
        In any edge-weighted bipartite subgraph $H$ of $G$, the maximum weight of a matching equals the smallest weight of a $w$-vertex-cover.
    \end{theorem}
    
    This theorem allows us to prove the following lemma, which is technically the most important part of the present work. 
    
    \begin{lemma} \label{lem:bipartite-w-edcs}
        Let $0 < \varepsilon < 1/2$ and $W$ be an integer parameter. For $\beta \geq \beta^- + 2$ integers such that $\frac{\beta}{\beta^-} \leq 1 + \frac{\varepsilon}{5W}$ and $\beta^- \geq \frac{4W}{\varepsilon}$, we have that any $(\beta, \beta^-)$-$w$-EDCS $H$ of a bipartite graph $G$ (with integer edge weights bounded by $W$) contains a matching $M_H$ such that $\left(2 - \frac{1}{2W} + \varepsilon\right)\cdot w(M_H) \geq w(M_G)$.
    \end{lemma}
    
    \begin{proof}
        Using K\H{o}nig-Egerváry's theorem in the bipartite graph $H$, we know that there exist integers $(\alpha_v)_{v \in V}$ such that:
        \begin{itemize}
            \item $\sum_{v \in V}\alpha_v = w(M_H)$
            \item for all $(u,v) \in H$, $w(u,v) \leq \alpha_u + \alpha_v$
        \end{itemize}
        Now consider the optimal matching $M_G$ in $G$. The first idea is to use the duality theorem to relate $w(M_G)$ to $w(M_H)$, with a leftover term that will be analyzed in the second part of the proof. 
        
        We introduce the notion of \emph{good} and \emph{bad} edges:
        \begin{itemize}
            \item the edges $(u,v) \in M_G$ such that $\beta^- \cdot w(u,v) \leq \beta \cdot (\alpha_u + \alpha_v)$, which are called \emph{good edges}; the set of good edges is denoted as $M_{good}$
            \item the edges $(u,v) \in M_G$ such that $\beta^- \cdot w(u,v) > \beta \cdot (\alpha_u + \alpha_v)$, which are called \emph{bad edges}; the set of bad edges is denoted as $M_{bad}$
        \end{itemize}
        
        A key observation is that the edges in $M_G \cap H$ are necessarily good edges by the definition of the $w$-vertex-cover $(\alpha_v)_{v \in V}$ and the fact that $\beta^- < \beta$. Therefore, the bad edges $(u,v)$ are in $G \backslash H$ and as a consequence they satisfy Property~(ii) of Definition~\ref{def:w-edcs}, \emph{i.e.} $\beta^- \cdot w(u,v) \leq \wdeg_H(u) + \wdeg_H(v)$.
        
        Hence we can write the following:
        \begin{align*}
            \beta^- &\cdot w(M_G) = \sum_{(u,v) \in M_{good}}\beta^- \cdot w(u,v) + \sum_{(u,v) \in M_{bad}}\beta^- \cdot w(u,v)\\
            &\leq \sum_{(u,v) \in M_{good}}\beta \cdot (\alpha_u + \alpha_v) + \sum_{(u,v) \in M_{bad}}(\wdeg_H(u) + \wdeg_H(v))\\
            &= \sum_{(u,v) \in M_G} \beta \cdot (\alpha_u + \alpha_v) + \sum_{(u,v) \in M_{bad}}(\wdeg_H(u) + \wdeg_H(v) - \beta \cdot (\alpha_u + \alpha_v))\\
            &\leq \beta \cdot w(M_H) + \sum_{(u,v) \in M_{bad}}((\wdeg_H(u) - \beta \cdot \alpha_u)_+ + (\wdeg_H(v) - \beta \cdot \alpha_v)_+),
        \end{align*}
        where $(x)_+$ denotes the non-negative part $\max(x, 0)$. In the last inequality we also used the fact that $\sum_{(u,v) \in M_G} (\alpha_u + \alpha_v) \leq \sum_{v \in V}\alpha_v = w(M_H)$, as each vertex of $V$ is counted at most once in the left-hand sum (because $M_G$ is a matching). Now, denoting by $V_{bad}$ the set of vertices which are the endpoints of a bad edge and such that $\wdeg_H(u) - \beta \cdot \alpha_u > 0$, we get
        \begin{equation}\beta^- \cdot w(M_G) \leq \beta \cdot w(M_H) + \sum_{v \in V_{bad}} (\wdeg_H(v) - \beta \cdot \alpha_v). 
        \label{equ:MGMH}
        \end{equation}
        
        Naturally, we want to upper-bound the value of $\sum_{v \in V_{bad}} (\wdeg_H(v) - \beta \cdot \alpha_v)$ and we will do so \emph{via} a specially-constructed graph. Before we describe this graph, we can first easily 
        observe that for any $v \in V_{bad}$, for any $u \in N_H(v)$, we have $w(u,v) \geq \frac{\wdeg_H(v)}{\beta} > \alpha_v$ (by Property~(i) of Definition~\ref{def:w-edcs} and the definition of $V_{bad}$); moreover, as $(\alpha_v)_{v \in V}$ is a $w$-vertex-cover of $H$, we obtain that $\alpha_u > 0$. These observations 
        will be useful in the following. 
        
        The new graph is $H_{bad} = (V_{bad} \cup \tilde{V}, \tilde{E})$. The vertices in $H_{bad}$ are the vertices of $V_{bad}$ as well as copies of the vertices of $V$ such that $\alpha_v > 0$, \emph{i.e.} $\tilde{V} = \{\tilde{v}: v \in V, \alpha_v > 0\}$. We build the set of edges $\tilde{E}$ as follows. For each $v \in V_{bad}$, for each $u \in N_H(v)$, we create in $\tilde{E}$ an edge $(v, \tilde{u})$ such that $w(v, \tilde{u}) = w(v,u) - \alpha_v$ (note that if $u$ is also in $V_{bad}$, then $\tilde{E}$ will also contain another edge $(u,\tilde{v})$ such that $w(u, \tilde{v}) = w(u,v) - \alpha_u$). 
        Note that $w(v, \tilde{u}) \in \mathbb{Z}_{> 0}$, since $w(u,v) > \alpha_v$ as observed above. Therefore the graph $H_{bad}$ still has non-negative 
        integer-valued edge weights. We next remove some edges from $\tilde{E}$: 
        while there exists a vertex $v \in V_{bad}$ such that $\wdeg_{H_{bad}}(v) > \wdeg_H(v) - \beta \cdot \alpha_v + W$, we pick an arbitrary edge $(v, \tilde{u}) \in \tilde{E}$ incident to $v$ and remove it from $H_{bad}$. 
        This process guarantees the following property: 
        \begin{equation} 
        \forall v \in V_{bad}, \wdeg_H(v) - \beta \cdot \alpha_v \leq \wdeg_{H_{bad}}(v) \leq \wdeg_H(v) - \beta \cdot \alpha_v + W.
        \label{equ:HbadWeightedDegree}
        \end{equation}
        
        This finishes the description of the graph $H_{bad}$. By~(\ref{equ:HbadWeightedDegree}), for any $(v, \tilde{u}) \in \tilde{E}$ we have:
        \begin{align*}
            \beta \cdot w(v, \tilde{u}) + W = \beta \cdot (w(v,u) - \alpha_v) + W
            &\geq \wdeg_H(v) - \beta \cdot \alpha_v + W + \wdeg_H(u)\\
            &\geq \wdeg_{H_{bad}}(v) + \wdeg_{H_{bad}}(\tilde{u}).
        \end{align*}
        Summing this inequality over all the edges in $\tilde{E}$ we obtain:
        \begin{align*}
            \beta &\cdot w(\tilde{E}) + W \cdot |\tilde{E}| \geq \sum_{(v,\tilde{u}) \in \tilde{E}}(\wdeg_{H_{bad}}(v) + \wdeg_{H_{bad}}(\tilde{u}))\\
            &= \sum_{v \in V_{bad}}\deg_{H_{bad}}(v) \cdot \wdeg_{H_{bad}}(v) + \sum_{\tilde{u} \in \tilde{V}}\deg_{H_{bad}}(\tilde{u}) \cdot \wdeg_{H_{bad}}(\tilde{u})\\
            &\geq \sum_{v \in V_{bad}}\frac{\wdeg_{H_{bad}}(v)}{W} \cdot \wdeg_{H_{bad}}(v) + \sum_{\tilde{u} \in \tilde{V}}\frac{\wdeg_{H_{bad}}(\tilde{u})}{\alpha_u} \cdot \wdeg_{H_{bad}}(\tilde{u})\\
            &= \sum_{v \in V_{bad}}\frac{(\wdeg_{H_{bad}}(v))^2}{W} + \sum_{\tilde{u} \in \tilde{V}}\frac{(\wdeg_{H_{bad}}(\tilde{u}))^2}{\alpha_u}\\
            &\geq \sum_{v \in V_{bad}}\frac{1}{W}\cdot \left(\frac{w(\tilde{E})}{|V_{bad}|}\right)^2 + \sum_{\tilde{u} \in \tilde{V}}\frac{1}{\alpha_u}\cdot \left(\frac{w(\tilde{E}) \cdot \alpha_u}{\sum_{\tilde{u}' \in \tilde{V}} \alpha_{u'}}\right)^2\\
            &= \frac{w(\tilde{E})^2}{W \cdot |V_{bad}|} + \frac{w(\tilde{E})^2}{\sum_{\tilde{u}' \in \tilde{V}} \alpha_{u'}}.
        \end{align*}
        The second inequality comes from the fact that the degree of a vertex can be lower-bounded by the weighted degree of that vertex divided by the weight of the largest edge incident to it (for $v \in V_{bad}$ this weight is $W$, and for $\tilde{u} \in \tilde{V}$ it is $\alpha_u$, as $w(v, \tilde{u}) = w(v,u) - \alpha_v \leq \alpha_u$ for $v$ adjacent to $\tilde{u}$ in $H_{bad}$). The third inequality comes from the minimization of the function over the constraints $\sum_{v \in V_{bad}}\wdeg_{H_{bad}}(v) = \sum_{\tilde{u} \in \tilde{V}}\wdeg_{H_{bad}}(\tilde{u})=w(\tilde{E})$. Then we get
        \begin{equation}\beta + W \geq \frac{w(\tilde{E})}{W \cdot |V_{bad}|} + \frac{w(\tilde{E})}{\sum_{\tilde{u} \in \tilde{V}} \alpha_{u}}.
        \label{equ:betaW}
        \end{equation} 
        
        The following claim will help us lower bound the average weighted degree of the vertices of $V_{bad}$ in $H_{bad}$, namely, $w(\tilde{E})/|V_{bad}|$. For this part it is crucial that the weights are integers.
        
        \begin{claim} \label{claim:average}
            For all $(u,v) \in M_{bad}$,
            $(\wdeg_H(u) - \beta \cdot \alpha_u)_+ + (\wdeg_H(v) - \beta \cdot \alpha_v)_+ \geq \frac{\beta^-}{1 + \varepsilon/4}$
        \end{claim}
        
        \begin{proof}
        We proceed by contradiction. 
            Suppose that there exists $(u,v) \in M_{bad}$ such that $(\wdeg_H(u) - \beta \cdot \alpha_u)_+ + (\wdeg_H(v) - \beta \cdot \alpha_v)_+ < \frac{\beta^-}{1 + \varepsilon/4}$. Then, as $\beta^- \cdot w(u,v) \leq \beta \cdot (\alpha_u + \alpha_v) + (\wdeg_H(u) - \beta \cdot \alpha_u)_+ + (\wdeg_H(v) - \beta \cdot \alpha_v)_+$, it means that
            \[\beta \cdot (\alpha_u + \alpha_v) < \beta^- \cdot w(u,v) < \beta \cdot (\alpha_u + \alpha_v) + \frac{\beta^-}{1+\varepsilon/4}, \]
            and therefore by dividing by $\beta^-$ we obtain
            \[\frac{\beta}{\beta^-} \cdot (\alpha_u + \alpha_v) < w(u,v) < \frac{\beta}{\beta^-}  \cdot (\alpha_u + \alpha_v) + \frac{1}{1+\varepsilon/4}.\]
            As $(\alpha_u + \alpha_v)\in \{0,1,\dots, W\}$ (recall that $(u,v)$ is a bad edge) and because $\frac{\beta}{\beta^-} \leq 1 + \frac{\varepsilon}{5W} < 1 + \frac{\varepsilon}{4W \cdot (1 + \varepsilon/4)}$, there cannot be any integer in the open interval \[\left]\frac{\beta}{\beta^-} \cdot (\alpha_u + \alpha_v),\, \frac{\beta}{\beta^-}  \cdot (\alpha_u + \alpha_v) + \frac{1}{1+\varepsilon/4}\right[,\]
            implying that $w(u,v)$, which is an integer, cannot exist. The proof follows. 
        \end{proof}
        
        Recall that $u$ of $(u,v) \in M_{bad}$
        is part of $V_{bad}$ only if 
        $\wdeg_H(u) - \beta \cdot \alpha_u > 0$. 
        Claim~\ref{claim:average} then implies that given $(u,v) \in M_{bad}$, if both $u$ and $v$ are in $V_{bad}$, then $\wdeg_{H_{bad}}(u) + \wdeg_{H_{bad}}(v) \geq \frac{\beta^-}{1 + \varepsilon/4}$; if only $u$ is in $V_{bad}$, 
        then $\wdeg_{H_{bad}}(u) \geq \frac{\beta^-}{1 + \varepsilon/4}$. 
        We can thus infer that  $\frac{w(\tilde{E})}{|V_{bad}|} \geq \frac{\beta^-}{2 \cdot (1 + \varepsilon/4)}$ and we can rewrite (\ref{equ:HbadWeightedDegree}) as $\beta + W \geq \frac{\beta^-}{2W \cdot (1 + \varepsilon/4)} + \frac{w(\tilde{E})}{\sum_{\tilde{u} \in \tilde{V}} \alpha_{u}},$
        and therefore
        \begin{equation}
        \label{equ:betaWNew}
        \left(\beta + W - \frac{\beta^-}{2W \cdot (1 + \varepsilon/4)}\right) \cdot \sum_{\tilde{u} \in \tilde{V}} \alpha_{u} \geq w(\tilde{E}).
        \end{equation}
        We now can rebound the expression of~(\ref{equ:MGMH}) as follows: 
        \begin{align*}
            \beta^- \cdot w(M_G) &\leq \beta \cdot w(M_H) + \sum_{v \in V_{bad}} (\wdeg_H(v) - \beta \cdot \alpha_v)& \\
            &\leq \beta \cdot w(M_H) + \sum_{v \in V_{bad}} \wdeg_{H_{bad}}(v) & \mbox{(using~(\ref{equ:HbadWeightedDegree}))}\\
            &\leq \beta \cdot w(M_H) + w(\tilde{E})& \\
            &\leq \left(2\beta + W - \frac{\beta^-}{2W \cdot (1 + \varepsilon/4)}\right) \cdot w(M_H). & \mbox{(using~(\ref{equ:betaWNew}))}
        \end{align*}
        Re-arranging, 
        \[\left(2\frac{\beta}{\beta^-} + \frac{W}{\beta^-} - \frac{1}{2W \cdot (1 + \varepsilon/4)}\right) \cdot w(M_H) \geq w(M_G).\]
        As $\frac{\beta}{\beta^-} \leq 1 + \varepsilon/4$ and $\beta^- \geq \frac{4W}{\varepsilon}$ we obtain the desired result. 
    \end{proof}
    
    Now we can generalize this result to non-bipartite graphs.
    
    \begin{theorem} \label{thm:general-w-edcs}
        Let $0 < \varepsilon < 1/2$ and $W$ be an integer parameter. Set $\lambda = \frac{\varepsilon}{100W}$. For $\beta \geq \beta^- + 2$ integers such that $\frac{\beta}{\log(\beta)} \geq 2 W^2\lambda^{-2}$ and $\beta^- \geq (1 - \lambda) \cdot \beta$, we have that any $(\beta, \beta^-)$-$w$-EDCS $H$ of a graph $G$ (with integer edge weights bounded by $W$) contains a matching $M_H$ such that $\left(2 - \frac{1}{2W} + \varepsilon\right)\cdot w(M_H) \geq w(M_G)$.
    \end{theorem}
    
    Here we use the probabilistic method and Lovasz Local Lemma~\cite{Lovasz1973}, as in~\cite{AssadiB19}.
        
    \begin{proposition}[Lovasz Local Lemma, see~\cite{Lovasz1973}] \label{prop:lovasz}
        Let $0 < p < 1$ and $d \geq 1$. Suppose $\mathcal{E}_1, \dots, \mathcal{E}_t$ are $t$ events such that $\mathbb{P}(\mathcal{E}_i) \leq p$ for all $i \in \{1, 2,\dots, t\}$ and each $\mathcal{E}_i$ is mutually independent of all but (at most) $d$ other events $\mathcal{E}_j$. If $(d+1) \cdot p < 1/e$ then the event $\bigcap_{i=1}^t\overline{\mathcal{E}_i}$ occurs with non-zero probability.
    \end{proposition}

    \label{app:proof-general-w-edcs}
    \begin{proof}[Proof of Theorem~\ref{thm:general-w-edcs}]
        Let $M_G$ be an optimal matching in $G$. Consider a randomly chosen bipartite subgraph $\tilde{G}(L \cup R, \tilde{E})$ of $G$ such that:
        \begin{itemize}
            \item for any edge $(u, v) \in M_G$, with probability $1/2$ we have $u \in L$ and $v \in R$, and with probability $1/2$ the opposite;
            \item for any vertex $v \in V$ unmatched by $M_G$, $v$ is assigned to $L$ or $R$ uniformly at random;
            \item all these random choices are made independently;
            \item the set of edges $\tilde{E}$ contains the edges of $G$ that have one endpoint in $L$ and the other endpoint in $R$, \emph{i.e.} $\tilde{E} = E \cap (L \times R)$.
        \end{itemize}
        Then we define $\tilde{H} = H \cap \tilde{G}$.
        
        Take any vertex $v \in V$. Let us assume that $v$ is chosen in $L$ in $\tilde{G}$ (the other case being symmetric). Then we have $\wdeg_{\tilde{H}}(v) = \sum_{u \in N_H(v)} \mathbbm{1}_{u \in R} \cdot w(u,v)$. Therefore, $\mathbb{E}[\wdeg_{\tilde{H}}(v)] = \wdeg_{H}(v)/2$ (plus an additional $w(v,v')/2$ term for some $v'$ if $(v, v') \in H$ is an edge of $M_G$). Moreover, if $u$ and $u'$ in $N_H(v)$ are matched in $M_G$ then the random variables $\mathbbm{1}_{u \in R} \cdot w(v,u)$ and $\mathbbm{1}_{u' \in R} \cdot w(v,u')$ and negatively associated. The other random variables $\mathbbm{1}_{u \in R} \cdot w(v,u)$ are mutually independent. Thereby, using Hoeffding's inequality (see Appendix~\ref{app:hoeffding}), we get
        \begin{align*}
            \mathbb{P}\left[\left|\wdeg_{\tilde{H}}(v) - \frac{\wdeg_{H}(v)}{2}\right| \geq \lambda \cdot \beta + \frac{W}{2}\right] &\leq \mathbb{P}\left[\left|\wdeg_{\tilde{H}}(v) - \mathbb{E}[\wdeg_{\tilde{H}}(v)]\right| \geq \lambda \cdot \beta\right]\\
            &\leq 2\cdot \exp\left(-\frac{2\lambda^2\cdot \beta^2}{\deg_H(v) \cdot W^2}\right)\\
            &\leq 2\cdot \exp\left(-\frac{2\lambda^2\cdot \beta^2}{\beta \cdot W^2}\right)\\
            &\leq 2\cdot \exp\left(-4 \cdot \log \beta \right) = \frac{2}{\beta^4}
        \end{align*}
        where the first inequality comes from the inequality $|\wdeg_{H}(v)/2 - \mathbb{E}[\wdeg_{\tilde{H}}(v)]| \leq W/2$, the third inequality comes from $\deg_H(v) \leq \beta$ (see Proposition~\ref{prop:w-edcs-bounded-degree}), and the fourth one from $\beta \geq 2  \cdot W^2 \cdot \lambda^{-2} \cdot \log\beta$.
        
        Now define as $\mathcal{E}_v$ the event that $\left|\wdeg_{\tilde{H}}(v) - \frac{\wdeg_{H}(v)}{2}\right| \geq \lambda \cdot \beta + W/2$. Note that $\mathcal{E}_v$ only depends on the choices for vertices in $N_H(v)$ and therefore can only depend on at most $\beta^2$ other events $\mathcal{E}_u$ for vertices $u$ which are neighbours to $N_H(v)$ in $H$ (as $\deg_H(u) \leq \beta$ for all $u \in V$, see Proposition~\ref{prop:w-edcs-bounded-degree}). Hence we can apply Lovasz Local Lemma (Proposition~\ref{prop:lovasz}) to prove that the event $\bigcap_{v \in V}\overline{\mathcal{E}_v}$ occurs with non-zero probability. Then, conditioning on this event, we can show that $\tilde{H}$ is a $(\tilde{\beta},\tilde{\beta}^-)$-$w$-EDCS of $\tilde{G}$ for some parameters $\tilde{\beta}$ and $\tilde{\beta}^-$.
        
        First, to prove Property~(i) of $w$-EDCS, let $(u,v)$ be any edge in $\tilde{H}$. By events $\overline{\mathcal{E}_u}$ and $\overline{\mathcal{E}_v}$ we have the inequality
        \begin{align*}
            \wdeg_{\tilde{H}}(u) + \wdeg_{\tilde{H}}(v) &\leq \frac{1}{2} \cdot (\wdeg_H(u) + \wdeg_H(v)) + 2\lambda \beta + W\\&
            \leq \beta/2\cdot w(u,v) + 2\lambda\beta + W\\
            &\leq \beta/2 \cdot (1 + 4 \lambda + 2W/\beta) \cdot w(u,v).
        \end{align*}
        For Property~(ii) of $w$-EDCS $\tilde{H}$, let $(u,v)$ be any edge in $\tilde{G}\backslash\tilde{H}$. By events $\overline{\mathcal{E}_u}$ and $\overline{\mathcal{E}_v}$ we have the inequality
        \begin{align*}
            \wdeg_{\tilde{H}}(u) + \wdeg_{\tilde{H}}(v) &\geq \frac{1}{2} \cdot (\wdeg_H(u) + \wdeg_H(v)) - 2\lambda \beta - W\\
            &\geq \beta^-/2\cdot w(u,v) - 2\lambda\beta - W\\
            &\geq \beta/2 \cdot (1 - 5 \lambda - 2W/\beta) \cdot w(u,v),
        \end{align*}
        as $\beta^- \geq (1 - \lambda) \cdot \beta$. Hence $\tilde{H}$ is a $(\tilde{\beta},\tilde{\beta}^-)$-$w$-EDCS of $\tilde{G}$ for $\tilde{\beta} = \lceil \beta/2 \cdot (1 + 4 \lambda + 2W/\beta) \rceil$ and $\tilde{\beta}^- = \lfloor \beta/2 \cdot (1 - 5 \lambda - 2W/\beta) \rfloor$. 
        
        The choices of $\lambda$, $\beta$, and $\beta^-$ guarantee that we can apply Lemma~\ref{lem:bipartite-w-edcs} to the subgraph $\tilde{H}$ in $\tilde{G}$. In fact, 
        \begin{align*}
            \tilde{\beta}^- &\geq \beta/2 \cdot (1 - 5 \lambda - 2W/\beta) - 1\\
            &\geq W^2 \lambda^{-2} \cdot (1 - 5 \lambda) - W - 1 &\text{as $\beta \geq 2W^2\lambda^{-2}$}\\
            &=\frac{100^2W^4}{\varepsilon^2} - \frac{500W^3}{\varepsilon} - W - 1 &\text{as $\lambda = \frac{\varepsilon}{100W}$}\\
            &\geq \frac{4W}{\varepsilon}
        \end{align*}
        and
        \begin{align*}
            \frac{\tilde{\beta}}{\tilde{\beta}^-} &\leq \frac{\beta/2 \cdot (1 + 4 \lambda + 2W/\beta) +1}{\beta/2 \cdot (1 - 5 \lambda - 2W/\beta) - 1}\\
            &= \frac{1 + 4 \lambda + 2W/\beta + 2/\beta}{1 - 5 \lambda - 2W/\beta - 2/\beta}\\
            &\leq \frac{1 + 4 \lambda + \lambda^2/W + \lambda^2/W^2}{1 - 5 \lambda - \lambda^2/W - \lambda^2/W^2}  &\text{as $\beta \geq 2W^2\lambda^{-2}$}\\
            &\leq \frac{1 + 6\lambda}{1 - 7 \lambda}  &\text{as $\lambda \leq 1$ and $W \geq 1$}\\
            &\leq 1 + 20\lambda &\text{as $\lambda \leq \frac{1}{20}$ (standard analysis)}\\
            &= 1 + \frac{\varepsilon}{5W} &\text{as $\lambda = \frac{\varepsilon}{100W}$}
        \end{align*}
        
        As $M_{\tilde{G}} = M_G$, $w(M_{\tilde{H}}) \leq w(M_H)$ (by construction), and $\tilde{H}$ is a $(\tilde{\beta},\tilde{\beta}^-)$-$w$-EDCS of $\tilde{G}$, it follows from Lemma~\ref{lem:bipartite-w-edcs} that $H$ contains a matching $M_H$ such that $\left(2 - \frac{1}{2W} + \varepsilon\right)\cdot w(M_H) \geq w(M_G)$.
    \end{proof}
    
\section{EDCS for Weighted $b$-Matchings}
    \label{sec:weighted-b-matching}

    From now on we consider the problem of finding a maximum weight $b$-matching in an edge-weighted multi-graph $G = (V,E)$. Hence we will use the notations described in the introduction. Here we recall the generalization of edge-degree constrained subgraphs (EDCS) to an edge-weighted multi-graph $G= (V, E)$ in the context of the $b$-matching problem.

    \begin{definition}
        \label{def:w-b-edcs}
        Let $G = (V,E)$ be a weighted multi-graph, $\{b_v\}_{v \in V}$ be a set of constraints, and $H$ be a subgraph of $G$. Given any integer parameters $\beta \geq 3$ and $\beta^- \leq \beta - 2$, we say that $H$ is a $(\beta, \beta^-)$-$w$-$b$-EDCS of $G$ if $H$ satisfies the following properties:
        \begin{enumerate}[(i)]
            \item \makebox[14em][l]{For any edge $(u, v, w_{uv}) \in H$,} $\frac{\wdeg_H(u)}{b_u} + \frac{\wdeg_H(v)}{b_v} \leq \beta \cdot w_{uv}$
            \item \makebox[14em][l]{For any edge $(u, v, w_{uv}) \in G \backslash H$,} $\frac{\wdeg_H(u)}{b_u} + \frac{\wdeg_H(v)}{b_v} \geq \beta^- \cdot w_{uv}$. 
        \end{enumerate}
    \end{definition}
    
    As for a $w$-EDCS (Proposition~\ref{prop:w-edcs-bounded-degree}), we can bound the degree of a vertex in a $w$-$b$-EDCS $H$.
    
    \begin{proposition}
        \label{prop:w-b-edcs-bounded-degree}
        Let $H$ be a $(\beta,\beta^-)$-$w$-$b$-EDCS of a given graph $G$. Then, for all $v \in V$, we have $\deg_H(v) \leq \beta \cdot b_v$.
    \end{proposition}
    
    \begin{proposition}
        \label{prop:w-b-edcs-bounded-size}
        Let $H$ be a $(\beta,\beta^-)$-$w$-$b$-EDCS of a given graph $G$. Then $H$ contains at most $2 \beta \cdot |M_G|$ edges.
    \end{proposition}
    
    \begin{proof}
        A vertex $v \in V$ is called \emph{saturated} by $M_G$ if $|\delta_G(v) \cap M_G| = b_v$. We denote by $V_{sat}$ the set of vertices saturated by $M_G$.  As $M_G$ is a maximal matching in $G$, it means that for all $(u,v, w_{uv}) \in G \backslash M_G$, either $u$ or $v$ is in $V_{sat}$.
        We denote by $M_{sat} \subseteq M_{G}$ the subset of edges in $M_G$ that are incident to a vertex of $V_{sat}$. By this definition, we get:
        \begin{align*}
            |H| &= |H \cap (M_G \backslash M_{sat})| + |H \backslash (M_G \backslash M_{sat})| \leq |M_G| - |M_{sat}| + \sum_{v \in V_{sat}} \deg_H(v)\\
            &\leq |M_G| - |M_{sat}| + \sum_{v \in V_{sat}}\beta \cdot b_v
            \leq |M_G| - |M_{sat}| + 2 \cdot |M_{sat}| \cdot \beta
            \leq 2\beta \cdot |M_G|
        \end{align*}
        as for all $v \in V$, $\deg_H(v) \leq \beta \cdot b_v$ and $\sum_{v \in V_{sat}} b_v \leq 2\cdot |M_{sat}|$.
    \end{proof}
    
    We can also show that a $w$-$b$-EDCS can always be extracted from a graph.
    \begin{proposition} \label{prop:exist-w-b-edcs}
        Any multi-graph $G = (V, E)$, along with a set of constraints $\{b_v\}_{v \in V}$, contains a $(\beta, \beta^-)$-$w$-$b$-EDCS for any parameters $\beta \geq \beta^- + 2$. Such a  $(\beta, \beta^-)$-$w$-$b$-EDCS can also be found in $O(\beta^2 W^2 \cdot |M_G|)$ time.
    \end{proposition}
    
    \begin{proof}
        As in the proof of Proposition~\ref{prop:exist-w-edcs}, we follow closely the argument of~\cite{AssadiB19}. We use the same local-search algorithm as the one in Proposition~\ref{prop:exist-w-edcs}, except that the properties violated are those of Definition~\ref{def:w-b-edcs}. Here we also give the priority to the correction of violations of Property~(i), so that the at each step of the algorithm all the vertices $v \in V$ have degrees bounded by $\beta \cdot b_v + 1$. To prove that this algorithm terminates and show the existence of a $w$-$b$-EDCS, we introduce the following  potential function:
        \[\Phi(H) = (2\beta - 2) \sum_{(u,v, w_{uv}) \in H}w_{uv}^2 - \sum_{u \in V} \frac{(\wdeg_H(u))^2}{b_u}.\]
        Observe that by the same argument used in Proposition~\ref{prop:w-b-edcs-bounded-size} (where we only use the fact that the degrees are bounded), the value of $\Phi(H)$ is bounded by $2\beta W^2 \cdot 2 \beta |M_G|$. We can also show that after each local improvement, the value of $\Phi(H)$ increases by at least $3/2$. 
        
        In fact, when Property~(i) is not respected by some edge $(u, v, w_{uv}) \in H$,  $\frac{\wdeg_H(u)}{b_u} + \frac{\wdeg_H(v)}{b_v} \geq \beta \cdot w_{uv} + \frac{1}{b_u \cdot b_v}$. If we fix it by erasing this edge from $H$ the value of $\Phi(H)$ is increased by $-(2\beta - 2) \cdot w_{uv}^2 + 2 \cdot \frac{\wdeg_H(u)}{b_u}\cdot w_{uv} + 2 \cdot \frac{\wdeg_H(v)}{b_v} \cdot w_{uv} - \frac{1}{b_u} \cdot w_{uv}^2 - \frac{1}{b_v} \cdot w_{uv}^2 \geq -(2\beta - 2) \cdot w_{uv}^2 + 2 \cdot (\beta \cdot w_{uv} + \frac{1}{b_u \cdot b_v}) \cdot w_{uv} - \frac{1}{b_u} \cdot w_{uv}^2 - \frac{1}{b_v} \cdot w_{uv}^2 = (2 - \frac{1}{b_u} - \frac{1}{b_v}) \cdot w_{uv}^2 + \frac{2}{b_u \cdot b_v}\cdot w_{uv} \geq 2 - \frac{1}{b_u} - \frac{1}{b_v} + \frac{2}{b_u \cdot b_v} \geq 3/2$ (studying the function $(x,y) \rightarrow 2 + 2\cdot xy - x - y$ on the domain $[0,1]^2$).
            
        When Property (ii) is not respected by some edge $(u, v, w_{uv}) \in G \backslash H$, it means that $\frac{\wdeg_H(u)}{b_u} + \frac{\wdeg_H(v)}{b_v} \leq \beta^- \cdot w_{uv} - \frac{1}{b_u \cdot b_v}$. If we insert $(u, v)$ in $H$, the value of $\Phi(H)$ increases by $(2\beta - 2) \cdot w_{uv}^2 - 2 \cdot \frac{\wdeg_H(u)}{b_u} \cdot w_{uv} - 2 \cdot \frac{\wdeg_H(v)}{b_v} \cdot w_{uv} - \frac{1}{b_u}\cdot w_{uv}^2 - \frac{1}{b_v} \cdot w_{uv}^2 \geq (2\beta - 2) \cdot w_{uv}^2 - 2 \cdot (\beta^- \cdot w_{uv} - \frac{1}{b_u \cdot b_v}) \cdot w_{uv} - \frac{1}{b_u}\cdot w_{uv}^2 - \frac{1}{b_v} \cdot w_{uv}^2 \geq (2 \cdot (\beta - \beta^-) - 2 - \frac{1}{b_u} - \frac{1}{b_v}) \cdot w_{uv}^2 + \frac{2}{b_u \cdot b_v} \cdot w_{uv} \geq (2 - \frac{1}{b_u} - \frac{1}{b_v}) \cdot w_{uv}^2 + \frac{2}{b_u \cdot b_v} \cdot w_{uv} \geq 3/2$, as before and because $\beta \geq \beta^- + 2$.
        
        Hence the algorithm terminates in $O(\beta^2W^2 \cdot |M_G|)$ steps.
    \end{proof}
    
    The main interest of these $w$-$b$-EDCSes is that they contain an (almost) $2 - \frac{1}{2W}$ approximation, as in the case of $w$-EDCSes in simple graphs (Theorem~\ref{thm:general-w-edcs}).
    
    \generalwbedcs*
    
    \begin{proof}
        Consider a maximum weight $b$-matching $M_G$. We build from $H$ and $M_G$ new multi-graphs $H' = (V', E_H')$ and $G' = (V', E')$ as follows. The set of vertices $V'$ contains, for each vertex $v \in V$, $b_v$ vertices $v_1, \dots, v_{b_v}$, so that $V'$ contains $\sum_{v \in V}b_v$ vertices in total. To construct $E'$, for each $v \in V$, we will distribute the edges of $\delta(v) \cap (H \cup M_G)$ among the $b_v$ vertices $v_1, \dots, v_{b_v}$ in such a way so that the following three properties hold: 
        \begin{enumerate}[(i)]
            \item each $v_i$ has a most one edge of $M_G$ incident to it;
            \item $G'$ is a simple graph (\emph{i.e.} with no multiple edges between two given vertices);
            \item each $v_i$ has a weighted degree in the interval $\left[\frac{\wdeg_H(v)}{b_v} - 2W, \frac{\wdeg_H(v)}{b_v} + 3W \right]$.
        \end{enumerate}
        The existence of such a distribution is guaranteed by Lemma~\ref{lem:distrib-edges} (see Appendix~\ref{app:proof-distrib-edges}). For the second property, it is crucial that the graph $G$ has at most $\min(b_u,b_v)$ edges between any vertices $u$ and $v$. This property is important in the proof of Theorem~\ref{thm:general-w-edcs} (when negative association is used). Then, for $H'$, we just consider the restriction of $G'$ to the edges corresponding to $H$ (ignoring those from $M_G \backslash H$ in the preceding construction).
        
        Observe that $M_G$ corresponds a simple matching in $G'$, and that any simple matching in $H'$ corresponds to a $b$-matching in $H$. Next we show that $H'$ is a $(\beta +6W, \beta^- - 6W)$-$w$-EDCS in the simple graph $G'$. Consider an edge $(u_i, v_j) \in H'$. It corresponds to an edge $(u,v, w_{uv})$ of $H$ so $\wdeg_{H'}(u_i) + \wdeg_{H'}(v_j) \leq \frac{\wdeg_H(u)}{b_u} + \frac{\wdeg_H(v)}{b_v} + 6W \leq (\beta  + 6W) \cdot w_{uv}$, so Property (i) of Definition~\ref{def:w-b-edcs} holds. Consider next an edge $(u_i, v_j) \in G' \backslash H'$. It corresponds to an edge $(u, v, w_{uv})$ of $M_G \backslash H$, so $\wdeg_{H'}(u_i) + \wdeg_{H'}(v_j) \geq \frac{\wdeg_H(u)}{b_u} + \frac{\deg_H(v)}{b_v} - 6W \geq (\beta^- - 6W) \cdot w_{uv}$ (as there can be a difference of at most $b_u \cdot W$ between the weighted degree of $u$ in $G'$ and in $H'$). Thus Property (ii) of Definition~\ref{def:w-b-edcs} holds as well. To conclude, $H'$ is a $(\beta +6W, \beta^- - 6W)$-$w$-EDCS of $G'$, so by Theorem~\ref{thm:general-w-edcs} we have that $(2 - \frac{1}{2W} + \varepsilon) \cdot w(M_{H'}) \geq w(M_{G'}) = w(M_G)$. As $w(M_H) \geq w(M_{H'})$ (because any matching in $H'$ corresponds to a $b$-matching of the same weight in $H$), we complete the proof. 
    \end{proof}

\section{Application to $b$-Matchings in Random-Order Streams}

    \label{sec:random-order}
    In this section we consider the random-order semi-streaming model and we 
    show how our results in the preceding section can be adapted 
    to get a $2 - \frac{1}{2W} + \varepsilon$ approximation.
    
    As our algorithm builds on that of Bernstein~\cite{bernstein:LIPIcs:2020:12419} for the unweighted simple matching, let us briefly summarize his approach. In the first phase 
    of the streaming, he constructs a subgraph that satisfies only a weaker definition 
    of EDCS in Definition~\ref{def:intro-edcs} (only Property (i) holds). 
    In the second phase of the streaming, he collects the ``underfull'' edges, which are those edges that violate Property (ii).
    He shows that in the end, the union of the subgraph built in the first phrase and the underfull edges collected in the second phase, 
    with high probability, contains a $3/2+\varepsilon$ approximation and that the total memory used is in the order of $O(n \cdot \log n)$. As we will show below, this approach can be properly adapted to our context of edge-weighted $b$-matching. Our main technical challenge lies in the fact that unlike the simple matching, the size of $M_{G}$ can vary a lot. We need a ``guessing'' strategy to ensure that the required memory is proportional to $|M_G|$. 
    
    \begin{definition}
        We say that a graph $H$ has bounded weighted edge-degree $\beta$
        if for every edge $(u,v, w_{uv}) \in H$, $\frac{\wdeg_H(u)}{b_u} + \frac{\wdeg_H(v)}{b_v} \leq \beta \cdot w_{uv}$. 
    \end{definition}
    
    \begin{definition} 
        Let $G$ be any edge-weighted multi-graph, and let $H$ be a subgraph of $G$ with bounded weighted edge-degree $\beta$. For any parameter $\beta^-$, we say that an edge $(u,v, w_{uv}) \in G \backslash H$ is $(H,\beta,\beta^-)$-underfull if  $\frac{\wdeg_H(u)}{b_u} + \frac{\wdeg_H(v)}{b_v} < \beta^- \cdot w_{uv}$.
    \end{definition}
    
    These two parts combined always contain a large matching:
    
    \begin{lemma}
        \label{lem:underfull-matching}
        Let $0 < \varepsilon < 1/2$ be any parameter and $W$ be an integer parameter. Set $\lambda = \frac{\varepsilon}{100W}$. Suppose that $\beta^-$ and $\beta \geq \beta^- + 2$ are integers so that $\frac{\beta + 8W}{\log(\beta + 8W)} \geq 2 W^2\lambda^{-2}$ and $\beta^- - 6W \geq (1 - \lambda) \cdot (\beta + 8W)$. Then for any edge-weighted multi-graph $G$ with integer weights in ${1,\dots,W}$, and for any subgraph $H$ with bounded weighted edge-degree $\beta$, if $X$ contains all edges in $G \backslash H$ that are $(H,\beta,\beta^-)$-underfull, then $(2 - \frac{1}{2W} + \varepsilon) \cdot w(M_{H \cup X}) \geq w(M_G)$.
    \end{lemma}
    
    \begin{proof}
        First, observe that $H \cup X$ is not necessarily a $w$-$b$-EDCS of $G$. Thereby we use another argument coming from~\cite{bernstein:LIPIcs:2020:12419}. Let $M_G$ be a maximum weight $b$-matching in $G$, let $M_G^H = M_G \cap H$ and $M_G^{G \backslash H} = M_G \cap (G \backslash H)$. Let $X^M = X \cap M_G^{G \backslash H}$. Clearly we have $w(M_G) = w(M_{H \cup M_G^{G \backslash H}})$. Then we can show that $H \cup X^M$ is a $(\beta + 2W, \beta^-)$-$w$-$b$-EDCS of $H \cup M_G^{G \backslash H}$.
        
        To show that $H \cup X^M$ is a $(\beta + 2W, \beta^-)$-$w$-$b$-EDCS of $H \cup M_G^{G \backslash H}$, a first observation is that for all $v \in V$, we have the inequalities $\wdeg_H(v) \leq \wdeg_{H \cup X^M}(v) \leq \wdeg_H(v) + b_v \cdot W$. We start with Property~(ii) of Definition~\ref{def:w-b-edcs}. By construction, $X^M$ contains all edges $(u, v, w_{uv})$ in $M_G^{G \backslash H}$, where the inequality $\frac{\wdeg_H(u)}{b_u} + \frac{\wdeg_H(v)}{b_v} < \beta^- \cdot w_{uv}$ is satisfied. Therefore, the remaining edges $(u,v, w_{uv}) \in (H \cup M_G^{G \backslash H}) \backslash (H \cup X^M) = M_G^{G \backslash H} \backslash X^M$ satisfy the inequality $\frac{\wdeg_H(u)}{b_u} + \frac{\wdeg_H(v)}{b_v} \geq \beta^- \cdot w_{uv}$. For Property~(i), for $(u,v,w_{uv}) \in H$, we have $\frac{\wdeg_{H \cup X^M}(u)}{b_u} + \frac{\wdeg_{H \cup X^M}(v)}{b_v} \leq (\beta + 2W)\cdot w_{uv}$. And for $(u, v, w_{uv}) \in X^M$, $\frac{\wdeg_{H \cup X^M}(u)}{b_u} + \frac{\wdeg_{H \cup X^M}(v)}{b_v} < (\beta^- + 2W)\cdot w_{uv} < (\beta + 2W)\cdot w_{uv}$, so Property~(i) is also satisfied.
        
        As a result, Theorem~\ref{thm:intro-general-w-b-edcs} can be applied in this case and we get that $\left(2 - \frac{1}{2W} + \varepsilon\right) \cdot w\left(M_{H \cup X^M}\right) \geq w\left(M_{H \cup M_G^{G \backslash H}}\right) = w(M_G)$, thus concluding the proof.
    \end{proof}

    \begin{remark} \label{rem:magnitude-beta}
        One can easily notice that there exist integers $\beta$ and $\beta^-$ that are $O(poly(W,1/\varepsilon))$ satisfying the conditions of Lemma~\ref{lem:underfull-matching}. From now on, we will use the parameters $\lambda$, $\beta$, and $\beta^-$ that satisfy those conditions and that are $O(poly(W,1/\varepsilon))$.
    \end{remark}
   
    \begin{algorithm}[h]
	\caption{Main algorithm computing a weighted $b$-matching for a random-order stream}\label{algo:w-b-matching}
	\begin{algorithmic}[1]
	\State $H \gets \emptyset$
	\State $\forall\,0 \leq i \leq \log_2 m,\, \alpha_i \gets \left\lfloor\frac{\varepsilon\cdot m }{\log_2(m) \cdot (2^{i+2}\beta^2W^2 + 1)}\right\rfloor$
	\For{$i = 0 \dots \log_2 m$}
	    \State $\textsc{ProcessStopped} \gets \textsc{False}$
	    \For{$2^{i+2}\beta^2W^2 + 1$ iterations}
	        \State $\textsc{FoundUnderfull} \gets \textsc{False}$
	        \For{$\alpha_i$ iterations}
	            \State let $(u,v, w_{uv})$ be the next edge in the stream
	            \If{$\frac{\wdeg_H(u)}{b_u} + \frac{\wdeg_H(v)}{b_v} < \beta^- \cdot w_{uv}$}
	                \State add edge $(u,v, w_{uv})$ to $H$
	                \State $\textsc{FoundUnderfull} \gets \textsc{True}$
                    \While{there exists $(u',v', w_{u'v'}) \in H : \frac{\wdeg_H(u')}{b_{u'}} + \frac{\wdeg_H(v')}{b_{v'}} > \beta \cdot w_{u'v'}$}
                        \State remove $(u', v', w_{u'v'})$ from $H$
                    \EndWhile
	            \EndIf
	        \EndFor
	        \If{$\textsc{FoundUnderfull} = \textsc{False}$}
	            \State $\textsc{ProcessStopped} \gets \textsc{True}$
	            \State \textbf{break} from the loop
	        \EndIf
	    \EndFor
	    \If{$\textsc{ProcessStopped} = \textsc{True}$}
	        \State \textbf{break} from the loop
	    \EndIf
	\EndFor
	\State $X \gets \emptyset$
	\For{each $(u,v,w_{uv})$ remaining edge in the stream}
	    \If{$\frac{\wdeg_H(u)}{b_u} + \frac{\wdeg_H(v)}{b_v} < \beta^- \cdot w_{uv}$}
	        \State add edge $(u,v, w_{uv})$ to $X$
	    \EndIf
	\EndFor
	\State \Return a maximum weight $b$-matching in $H \cup X$
	\end{algorithmic}
	\end{algorithm}
    
    The algorithm, formally described in Algorithm~\ref{algo:w-b-matching}, can be separated into two different phases of streaming. The first phase, corresponding to Lines~3-18, constructs some a subgraph $H$ of bounded weighted edge-degree $\beta$ using only a $\varepsilon$ fraction of the stream $E^{early}$. In the second phase, the algorithm collects the underfull edges in the remaining part of the stream $E^{late}$. As in~\cite{bernstein:LIPIcs:2020:12419} we use the idea that if no underfull edge was found in an interval of size $\alpha$ (see Lines~6-13), it means that with high probability the number of underfull edges remaining in the stream is bounded by some value $\gamma = 4 \log(m) \frac{m}{\alpha}$. The issue is therefore to choose the right size of interval $\alpha$, because we do not know the order of magnitude of $|M_G|$ in the $b$-matching problem: if we do as in~\cite{bernstein:LIPIcs:2020:12419} by choosing only one fixed size of intervals $\alpha$, then if $\alpha$ is too small, the value of $\gamma$ will be too big compared to $|M_G|$, whereas if the value of $\alpha$ is too large we will not be able to terminate the first phase of the algorithm within the early fraction of size $\varepsilon m$.
    Therefore, the idea in the first phase of the algorithm is to ``guess'' the value of $\log_2|M_G|$ by trying successively larger and larger values of $i$ (see Line~3). In fact, by setting $i_0 = \lceil \log_2 |M_G| \rceil$, we know that the number of operations that can be performed on a $w$-$b$-EDCS is bounded by $2^{i_0+2}\beta^2W^2$ (see the proof of Proposition~\ref{prop:exist-w-b-edcs}). As a result we know that the first phase should always stop at a time where $i$ is smaller or equal to $i_0$, and therefore at a time when $\alpha_i \geq \alpha_{i_0}$. Then we can prove that with high probability the number of remaining underfull edges in the stream is at most $\gamma_i = 4 \log(m) \frac{m}{\alpha_i}$.
	
	Algorithm~\ref{algo:w-b-matching} works when $M_G$ is neither too small nor too big. Here we will first argue that 
	the other border cases can be handled anyway. We first have this easy lemma (its proof is very similar to that of Proposition~\ref{prop:w-b-edcs-bounded-size}).
	
	\begin{lemma}
	    \label{lem:bound-graph-size}
	    We have the inequality $|G| \leq  2n \cdot |M_G|$.
	\end{lemma}
	
	\begin{proof}
	    We denote by $V_{sat}$ the set of vertices saturated by $M_G$ and $M_{sat} \subseteq M_G$ the subset of edges that are incident to a vertex of $V_{sat}$. 
	    Recall that as $M_G$ is a maximal matching in $G$, for all $(u,v, w_{uv}) \in G \backslash M_G$, either $u$ or $v$ is in $V_{sat}$. Therefore, 
	    \begin{align*}
	        |G| = |M_G \backslash M_{sat}| + |G \backslash (M_G \backslash M_{sat})| &\leq |M_G| - |M_{sat}| + \sum_{v \in V_{sat}} \deg_{G}(v)\\
            &\leq |M_G| - |M_{sat}| + \sum_{v \in V_{sat}}n \cdot b_v\\
            &\leq |M_G| - |M_{sat}| + 2 \cdot |M_{sat}| \cdot n\\
            &\leq 2n \cdot |M_G|,
	    \end{align*}
	    as for all $v \in V$, $\deg_{G}(v) \leq n \cdot b_v$ and $\sum_{v \in V_{sat}} b_v \leq 2\cdot |M_{sat}|$.
	\end{proof}
	
	Then we use it to handle the case of small $b$-matchings.
	
	\begin{claim} \label{claim:small-output}
	    We can assume that $w(M_G) \geq \frac{3W^2}{2\varepsilon^2}\log(m)$.
	\end{claim}
	
	\begin{proof}
	    In fact, if $w(M_G) < \frac{3W^2}{2\varepsilon^2}\log(m)$, then $|M_G| < \frac{3W^2}{2\varepsilon^2}\log(m)$ and by Lemma~\ref{lem:bound-graph-size} the graph has only $m = O(n \cdot \frac{3W^2}{2\varepsilon^2} \cdot \log(m))$ edges, so the whole graph can be stored so we can compute an exact solution only using $O(n \cdot poly(\log(m),W,1/\varepsilon))$ memory.
	\end{proof}
	
	\begin{claim} \label{claim:late-part}
	    Consider a maximum weight matching $M_G$. Assuming Claim~\ref{claim:small-output}, with probability at least $1 - m^{-3}$ the late part of the stream $E^{late}$ contains at least a $(1 - 2 \varepsilon)$ fraction of that optimal $b$-matching, \emph{i.e.}, $w(E^{late} \cap M_G) \geq (1 - 2 \varepsilon) \cdot w(M_G)$.
	\end{claim}
	
	\begin{proof}
	    Consider a maximum weight $b$-matching $M_G = \{f_1,\ldots, f_{|M_G|}\}$. We define the random variables $X_i = \mathbbm{1}_{f_i \in E^{early}} \cdot w(f_i)$. Hence we have $\mathbb{E}[\sum X_i] = \varepsilon \cdot w(M_G)$. Moreover, the random variables $X_i$ are negatively associated, so we can use Hoeffding's inequality (see Appendix~\ref{app:hoeffding}) to get
	    \[\mathbb{P}\left[\sum_{i = 1}^{|M_G|} X_i \geq 2 \varepsilon \cdot w(M_G)\right] \leq \exp\left(-\frac{2 \cdot \varepsilon^2 \cdot w(M_G)^2}{|M_G| \cdot W^2}\right) \leq \exp\left(-\frac{2 \cdot \varepsilon^2 \cdot w(M_G)}{W^2}\right) \leq m^{-3},\]
	    as we now assume that $w(M_G) \geq \frac{3W^2}{2\varepsilon^2}\log(m)$ (see Claim~\ref{claim:small-output}).
	\end{proof}
	
	Recall that we defined $i_0 = \lceil \log_2 |M_G| \rceil$. 
	
	\begin{claim} \label{claim:big-output}
	    We can assume that $\frac{\varepsilon\cdot m}{\log_2(m) \cdot (2^{i_0+2}\beta^2W^2 + 1)} \geq 1$.
	\end{claim}
	
	\begin{proof}
	    If this is not the case, then we can just store all the edges of $G$ as the number of edges $m$  is bounded by $\frac{\log_2(m) \cdot (2^{i_0+2}\beta^2W^2 + 1)}{\varepsilon} = O(|M_G|\cdot poly(\log(m),W,1/\varepsilon))$ (as $\beta$ is $O(poly(W,1/\varepsilon))$, see Remark~\ref{rem:magnitude-beta}). As a result, if at some point of the first phase we have not stopped and we have $\alpha_i = 0$, then we store all the remaining edges of $E^{late}$ and we will be able to get a $(1 - 2\varepsilon)$ approximation with high probability (because of Claim~\ref{claim:late-part}) using $O(|M_G|\cdot poly(\log(m),W,1/\varepsilon))$ memory.
	\end{proof}
	
	Then we can move to our main algorithm. The following lemma is very similar to the one used in~\cite{bernstein:LIPIcs:2020:12419} (see the proof in the Appendix). It can then be combined with previous lemmas and claims to prove that a $2 - \frac{1}{2W} + \varepsilon$ approximation can be achieved with high probability.
	
	\begin{lemma} \label{lem:early-properties}
	    The first phase of Algorithm~\ref{algo:w-b-matching} uses $O(\beta\cdot |M_G|)$ memory and constructs a subgraph $H$ of $G$, satisfying the following properties:
	    \begin{enumerate}
	        \item The first phase terminates within the first $\varepsilon m$ edges of the stream.
	        \item When the first phase terminates after processing some edge, we have:
	            \begin{enumerate}
	                \item $H$ has bounded weighted edge degree $\beta$, and  contains at most $O(\beta \cdot |M_G|)$ edges.
	                \item With probability at least $1 - m^{-3}$, the total number of $(H, \beta, \beta^-)$-underfull edges in the remaining part of the stream is at most $\gamma = O(|M_G|\cdot (\log(m))^2 \cdot \beta^2 W^2 \cdot 1/\varepsilon)$.
	            \end{enumerate}
	    \end{enumerate}
	\end{lemma}
	
	\begin{proof}
	    First, in each interval of size $\alpha_i$ processed until the first phase terminates (except the last interval) there is at least one insertion/deletion operation that is performed (as described in the proof of Proposition~\ref{prop:exist-w-b-edcs}) and therefore the total number of such processed intervals is bounded by $4 \beta^2W^2\cdot |M_G| + 1$. As a result, the first phase ends with some $i \leq i_0 = \lceil \log_2 |M_G| \rceil$, and the total number of edges processed in the first phase is therefore bounded by $\varepsilon m \cdot \frac{i_0}{\log_2(m)} \leq \varepsilon m$. For Property~2.a, as the subgraph $H$ built always keeps a bounded weighed edge-degree $\beta$, Proposition~\ref{prop:w-b-edcs-bounded-size} implies that  $H$ uses $O(\beta \cdot |M_G|)$ memory.

        Now we turn to the last property. As mentioned previously, the intuition is simple: the algorithm only exits the first phase if it fails to find a single underfull edge in an entire interval (Line~14-16), and since the stream is random, such an event implies that there are most likely few underfull edges left in the stream.

        To formalize this, we call the $j$-th time that Lines 7-13 are processed the \emph{epoch} $j$. Let $\mathcal{A}_j$ be the event that $\textsc{FoundUnderfull}$ is set to $\textsc{False}$ in epoch $j$. Let $\mathcal{B}_j$ be the event that the number of $(H,\beta, \beta^-)$-underfull edges in the remaining part of the stream is larger than some $\gamma$. Note that the last property fails to hold if and only if we have $\mathcal{A}_j \land \mathcal{B}_j$ for some $j$, so we want to upper bound $\mathbb{P}[\mathcal{A}_j \land \mathcal{B}_j]$. Let $E_j^r$ be the set of all edges in the graph that have not yet appeared in the stream at the \emph{beginning} of epoch $j$ ($r$ for remaining). Let $E^e_j$ be the edges that appear in epoch $j$ ($e$ for epoch), and note that $E^e_j$ is a subset of size $\alpha_i \geq \alpha_{i_0} = \alpha_{\lceil\log_2|M_G|\rceil} = \alpha$ chosen uniformly at random from $E^r_j$. Define $H_j$ to be the subgraph $H$ at the beginning of epoch $j$, and define $E^u_j \subseteq E^r_j$ to be the set of remaing underfull edges with respect to $H_j$, $\beta$, and $\beta^-$. Observe that because of event $\mathcal{A}_j$, the graph $H$ does not change throughout epoch $j$, so an edge that is underfull at any point during the epoch will be underfull at the end as well. Thus, $\mathcal{A}_j \land \mathcal{B}_j$ is equivalent to the event that $|E^u_j| > \gamma$ but $E^u_j \cap E^e_j = \emptyset$. 

        Let $\mathcal{A}^k_j$ be the event that the $k$-th edge of epoch $j$ is not in $E^u_j$. We have that \[\mathbb{P}[\mathcal{B}_j \land \mathcal{A}_j] \leq \mathbb{P}[\mathcal{A}_j \,|\, \mathcal{B}_j] \leq \mathbb{P}[\mathcal{A}^1_j \, | \, \mathcal{B}_j] \prod_{k=2}^\alpha \mathbb{P}[\mathcal{A}^k_j \, | \, \mathcal{B}_j, \mathcal{A}^1_j, \ldots, \mathcal{A}^{k-1}_j],\] the second inequality coming from the fact that $E^e_j$ is of size larger or equal to $\alpha = \alpha_{\lceil\log_2|M_G|\rceil}$.

        Now, observe that $\mathbb{P}[\mathcal{A}^1_j \,|\, \mathcal{B}_j] < 1 - \frac{\gamma}{m}$ because the first edge of the epoch is chosen uniformly at random from the set of $\leq m$ remaining edges, and the event fails if the chosen edge is in $E^u_j$, where $|E^u_j| > \gamma$ by definition of $\mathcal{B}_j$. Similarly, for any $k$,
        $\mathbb{P}[\mathcal{A}^k_j \, | \, \mathcal{B}_j, \mathcal{A}^1_j, \ldots, \mathcal{A}^{k-1}_j] < 1 - \frac{\gamma}{m}$ because conditioning on the previous events $\mathcal{A}^t_j$ implies that no edge from $E^u_j$ has yet appeared in this epoch, so there are still at least $\gamma$ edges from $E^u_j$ left in the stream.

        We now set 
        \[\gamma = 4\log(m)\cdot \frac{m}{\alpha} = 4\log(m)\cdot m \cdot \left\lfloor\frac{\varepsilon\cdot m }{\log_2(m) \cdot (2^{i_0+2}\beta^2W^2 + 1)}\right\rfloor^{-1},\] which is $O(|M_G|\cdot (\log(m))^2 \cdot \beta^2 W^2 \cdot 1/\varepsilon)$, as we assumed that $\frac{\varepsilon\cdot m }{\log_2(m) \cdot (2^{i_0+2}\beta^2W^2 + 1)} \geq 1$ (and there is at most a factor $2$ between $\lfloor x \rfloor$ and $x$ if $x \geq 1$).
        
        Combining the above equations yields \[\mathbb{P}[\mathcal{B}_j \land \mathcal{A}_j] \leq \left(1-\frac{\gamma}{m}\right)^{\alpha} = \left(1 - \frac{4\log(m)}{\alpha}\right)^\alpha \leq m^{-4}.\] There are clearly at most $m$ epochs, so union bounding over all of them shows that the last property fails with probability at most $m^{-3}$, as desired.
	\end{proof}
	
	\begin{theorem}
	    \label{thm:streaming-approx}
	    Let $\varepsilon > 0$. Using Algorithm~\ref{algo:w-b-matching}, with probability $1 - 2m^{-3}$, one can extract from a randomly-ordered stream of edges a weighted $b$-matching with an approximation ratio of $2 - \frac{1}{2W} + \varepsilon$, using $O(\max(|M_G|,n) \cdot poly(\log(m), W, 1/\varepsilon))$ memory.
	\end{theorem}
	
	\begin{proof}
	    Applying Lemma~\ref{lem:underfull-matching} to the graph $H \cup G^{late}$ we can get, choosing the right values $\beta$ and $\beta^-$ (which are $O(poly(W, 1/\varepsilon))$), $H \cup X$ contains a $(1 - 2\varepsilon)^{-1}\cdot (2 - \frac{1}{2W} + \varepsilon)$ approximation of the optimal $b$-matching (with probability at least $1 - m^{-3}$, see Claim~\ref{claim:late-part}), and with  a memory consumption of $O(|M_G| \cdot poly(\log(m), W, 1/\varepsilon))$ (with probability at least $1 - m^{-3}$, see Lemma~\ref{lem:early-properties}), with probability at least $1 - 2m^{-3}$ (union bound). Hence the proof.
	\end{proof}

\bibliography{library}

\begin{thebibliography}{10}

\bibitem{AssadiB21}
Sepehr Assadi and Soheil Behnezhad.
\newblock Beating two-thirds for random-order streaming matching.
\newblock In Nikhil Bansal, Emanuela Merelli, and James Worrell, editors, {\em
  48th International Colloquium on Automata, Languages, and Programming,
  {ICALP} 2021, July 12-16, 2021, Glasgow, Scotland (Virtual Conference)},
  volume 198 of {\em LIPIcs}, pages 19:1--19:13. Schloss Dagstuhl -
  Leibniz-Zentrum f{\"{u}}r Informatik, 2021.

\bibitem{AssadiB19}
Sepehr Assadi and Aaron Bernstein.
\newblock Towards a unified theory of sparsification for matching problems.
\newblock In Jeremy~T. Fineman and Michael Mitzenmacher, editors, {\em 2nd
  Symposium on Simplicity in Algorithms, {SOSA} 2019, January 8-9, 2019, San
  Diego, CA, {USA}}, volume~69 of {\em {OASICS}}, pages 11:1--11:20. Schloss
  Dagstuhl - Leibniz-Zentrum f{\"{u}}r Informatik, 2019.

\bibitem{bernstein:LIPIcs:2020:12419}
Aaron Bernstein.
\newblock Improved bounds for matching in random-order streams.
\newblock In Artur Czumaj, Anuj Dawar, and Emanuela Merelli, editors, {\em 47th
  International Colloquium on Automata, Languages, and Programming, {ICALP}
  2020, July 8-11, 2020, Saarbr{\"{u}}cken, Germany (Virtual Conference)},
  volume 168 of {\em LIPIcs}, pages 12:1--12:13. Schloss Dagstuhl -
  Leibniz-Zentrum f{\"{u}}r Informatik, 2020.

\bibitem{BernsteinDL21}
Aaron Bernstein, Aditi Dudeja, and Zachary Langley.
\newblock A framework for dynamic matching in weighted graphs.
\newblock In Samir Khuller and Virginia~Vassilevska Williams, editors, {\em
  {STOC} '21: 53rd Annual {ACM} {SIGACT} Symposium on Theory of Computing,
  Virtual Event, Italy, June 21-25, 2021}, pages 668--681. {ACM}, 2021.

\bibitem{BernsteinS15}
Aaron Bernstein and Cliff Stein.
\newblock Fully dynamic matching in bipartite graphs.
\newblock In Magn{\'{u}}s~M. Halld{\'{o}}rsson, Kazuo Iwama, Naoki Kobayashi,
  and Bettina Speckmann, editors, {\em Automata, Languages, and Programming -
  42nd International Colloquium, {ICALP} 2015, Kyoto, Japan, July 6-10, 2015,
  Proceedings, Part {I}}, volume 9134 of {\em Lecture Notes in Computer
  Science}, pages 167--179. Springer, 2015.

\bibitem{BernsteinS16}
Aaron Bernstein and Cliff Stein.
\newblock Faster fully dynamic matchings with small approximation ratios.
\newblock In Robert Krauthgamer, editor, {\em Proceedings of the Twenty-Seventh
  Annual {ACM-SIAM} Symposium on Discrete Algorithms, {SODA} 2016, Arlington,
  VA, USA, January 10-12, 2016}, pages 692--711. {SIAM}, 2016.

\bibitem{Egervary1931matrixok}
Jeno Egerv{\'a}ry.
\newblock Matrixok kombinatorius tulajdons{\'a}gair{\'o}l.
\newblock {\em Matematikai {\'e}s Fizikai Lapok}, 38(1931):16--28, 1931.

\bibitem{Lovasz1973}
Paul Erd{\H{o}}s and L{\'a}szl{\'o} Lov{\'a}sz.
\newblock Problems and results on 3-chromatic hypergraphs and some related
  questions.
\newblock In {\em Colloquia Mathematica Societatis Janos Bolyai 10. Infinite
  and Finite Sets, Keszthely (Hungary)}. Citeseer, 1973.

\bibitem{FarhadiHMRR20}
Alireza Farhadi, Mohammad~Taghi Hajiaghayi, Tung Mai, Anup Rao, and Ryan~A.
  Rossi.
\newblock Approximate maximum matching in random streams.
\newblock In Shuchi Chawla, editor, {\em Proceedings of the 2020 {ACM-SIAM}
  Symposium on Discrete Algorithms, {SODA} 2020, Salt Lake City, UT, USA,
  January 5-8, 2020}, pages 1773--1785. {SIAM}, 2020.

\bibitem{FKMSZ2005}
Joan Feigenbaum, Sampath Kannan, Andrew McGregor, Siddharth Suri, and Jian
  Zhang.
\newblock On graph problems in a semi-streaming model.
\newblock {\em Theor. Comput. Sci.}, 348(2-3):207--216, 2005.

\bibitem{GamlathSSS2019}
Buddhima Gamlath, Sagar Kale, Slobodan Mitrovic, and Ola Svensson.
\newblock Weighted matchings via unweighted augmentations.
\newblock In Peter Robinson and Faith Ellen, editors, {\em Proceedings of the
  2019 {ACM} Symposium on Principles of Distributed Computing, {PODC} 2019,
  Toronto, ON, Canada, July 29 - August 2, 2019}, pages 491--500. {ACM}, 2019.

\bibitem{GW2019}
Mohsen Ghaffari and David Wajc.
\newblock Simplified and space-optimal semi-streaming (2+epsilon)-approximate
  matching.
\newblock In Jeremy~T. Fineman and Michael Mitzenmacher, editors, {\em 2nd
  Symposium on Simplicity in Algorithms, {SOSA} 2019, January 8-9, 2019, San
  Diego, CA, {USA}}, volume~69 of {\em OASIcs}, pages 13:1--13:8. Schloss
  Dagstuhl - Leibniz-Zentrum f{\"{u}}r Informatik, 2019.

\bibitem{GuruganeshS17}
Guru~Prashanth Guruganesh and Sahil Singla.
\newblock Online matroid intersection: Beating half for random arrival.
\newblock In Friedrich Eisenbrand and Jochen K{\"{o}}nemann, editors, {\em
  Integer Programming and Combinatorial Optimization - 19th International
  Conference, {IPCO} 2017, Waterloo, ON, Canada, June 26-28, 2017,
  Proceedings}, volume 10328 of {\em Lecture Notes in Computer Science}, pages
  241--253. Springer, 2017.

\bibitem{HuangS21}
Chien{-}Chung Huang and Fran{\c{c}}ois Sellier.
\newblock Semi-streaming algorithms for submodular function maximization under
  b-matching constraint.
\newblock In Mary Wootters and Laura Sanit{\`{a}}, editors, {\em Approximation,
  Randomization, and Combinatorial Optimization. Algorithms and Techniques,
  {APPROX/RANDOM} 2021, August 16-18, 2021, University of Washington, Seattle,
  Washington, {USA} (Virtual Conference)}, volume 207 of {\em LIPIcs}, pages
  14:1--14:18. Schloss Dagstuhl - Leibniz-Zentrum f{\"{u}}r Informatik, 2021.

\bibitem{Huang_et_al:LIPIcs.ESA.2022.68}
Chien{-}Chung Huang and Fran{\c{c}}ois Sellier.
\newblock Maximum weight b-matchings in random-order streams.
\newblock In Shiri Chechik, Gonzalo Navarro, Eva Rotenberg, and Grzegorz
  Herman, editors, {\em 30th Annual European Symposium on Algorithms, {ESA}
  2022, September 5-9, 2022, Berlin/Potsdam, Germany}, volume 244 of {\em
  LIPIcs}, pages 68:1--68:14. Schloss Dagstuhl - Leibniz-Zentrum f{\"{u}}r
  Informatik, 2022.

\bibitem{Kapralov2021}
Michael Kapralov.
\newblock Space lower bounds for approximating maximum matching in the edge
  arrival model.
\newblock In D{\'{a}}niel Marx, editor, {\em Proceedings of the 2021 {ACM-SIAM}
  Symposium on Discrete Algorithms, {SODA} 2021, Virtual Conference, January 10
  - 13, 2021}, pages 1874--1893. {SIAM}, 2021.

\bibitem{Konrad18}
Christian Konrad.
\newblock A simple augmentation method for matchings with applications to
  streaming algorithms.
\newblock In Igor Potapov, Paul~G. Spirakis, and James Worrell, editors, {\em
  43rd International Symposium on Mathematical Foundations of Computer Science,
  {MFCS} 2018, August 27-31, 2018, Liverpool, {UK}}, volume 117 of {\em
  LIPIcs}, pages 74:1--74:16. Schloss Dagstuhl - Leibniz-Zentrum f{\"{u}}r
  Informatik, 2018.

\bibitem{KonradMM12}
Christian Konrad, Fr{\'{e}}d{\'{e}}ric Magniez, and Claire Mathieu.
\newblock Maximum matching in semi-streaming with few passes.
\newblock In Anupam Gupta, Klaus Jansen, Jos{\'{e}} D.~P. Rolim, and Rocco~A.
  Servedio, editors, {\em Approximation, Randomization, and Combinatorial
  Optimization. Algorithms and Techniques - 15th International Workshop,
  {APPROX} 2012, and 16th International Workshop, {RANDOM} 2012, Cambridge, MA,
  USA, August 15-17, 2012. Proceedings}, volume 7408 of {\em Lecture Notes in
  Computer Science}, pages 231--242. Springer, 2012.

\bibitem{LW2021}
Roie Levin and David Wajc.
\newblock Streaming submodular matching meets the primal-dual method.
\newblock In D{\'{a}}niel Marx, editor, {\em Proceedings of the 2021 {ACM-SIAM}
  Symposium on Discrete Algorithms, {SODA} 2021, Virtual Conference, January 10
  - 13, 2021}, pages 1914--1933. {SIAM}, 2021.

\bibitem{PS2017}
Ami Paz and Gregory Schwartzman.
\newblock A (2+{\(\epsilon\)})-approximation for maximum weight matching in the
  semi-streaming model.
\newblock volume~15, pages 18:1--18:15, 2019.

\bibitem{tennenholtz2000some}
Moshe Tennenholtz.
\newblock Some tractable combinatorial auctions.
\newblock In Henry~A. Kautz and Bruce~W. Porter, editors, {\em Proceedings of
  the Seventeenth National Conference on Artificial Intelligence and Twelfth
  Conference on on Innovative Applications of Artificial Intelligence, July 30
  - August 3, 2000, Austin, Texas, {USA}}, pages 98--103. {AAAI} Press / The
  {MIT} Press, 2000.

\end{thebibliography}

\appendix

\section{Details from Earlier Sections}

    \subsection{Hoeffding's Inequality}
        \label{app:hoeffding}
        We recall the following probabilistic tool we will use in this paper.
        \begin{proposition}[Hoeffding's inequality] \label{prop:hoeffding}
            Let $X_1, \dots, X_t$ be $t$ negatively associated random variables that take values in $[0,W]$. Let $X := \sum_{i = 1}^tX_i$. Then, for all $\lambda > 0$ we have:
            \[\mathbb{P}(X - \mathbb{E}[X] \geq \lambda) \leq \exp\left(-\frac{2\lambda^2}{t \cdot W^2}\right) \,\text{ and }\, \mathbb{P}(|X - \mathbb{E}[X]| \geq \lambda) \leq 2 \cdot \exp\left(-\frac{2\lambda^2}{t \cdot W^2}\right).\]
        \end{proposition}
        
    \subsection{Existence of the Edge Distribution for Theorem~\ref{thm:intro-general-w-b-edcs}}
        \label{app:proof-distrib-edges}
        \begin{lemma} \label{lem:distrib-edges}
            In the proof of Theorem~\ref{thm:intro-general-w-b-edcs}, for $v \in V$, the edges of $\delta_G(v) \cap (H \cup M_G)$ can be distributed among the $b_v$ vertices $v_1, \dots, v_{b_v}$ so that:
            \begin{enumerate}[(i)]
                \item each $v_i$ has a most one edge of $M_G$ incident to it;
                \item $G'$ is a simple graph (\emph{i.e.} with no multiple edge between two given vertices);
                \item each $v_i$ has a weighted degree in the interval $\left[\frac{\wdeg_H(v)}{b_v} - 2W, \frac{\wdeg_H(v)}{b_v} + 3W \right]$.
            \end{enumerate}
        \end{lemma}
        
        \begin{proof}
            
        Fix $v \in V$. 
        Let us introduce some notation. We write $S = \delta_G(v) \cap (H \cup M_G)$. For all $u \in V$, $u \neq v$, we denote by $S_u$ the multi-set of weighted edges between $u$ and $v$ in $H \cup M_G$, \emph{i.e.} $S_u = \delta_G(u) \cap \delta_G(v) \cap (H \cup M_G)$. As $|\delta_G(u) \cap \delta_G(v)| \leq \min(b_u,b_v)$, $|S_u| \leq b_v$. Setting $s_u=|S_u|$, we denote by $S_u = \{e_{u,1},\ldots,e_{u,s_u}\}$ the multi-set of weighted edges between $u$ and $v$, assuming that $w(e_{u,1}) \geq \ldots \geq w(e_{u,s_u})$. Setting $m_u = |M_G \cap S_u|$, we can assume that these $m_u$ edges from $M_G$ are $e_{u,1},\ldots, e_{u,m_u}$ (as these are the heaviest edges in $S_u$).
            
        Hence rephrasing the desired result, we want to distribute the edges of $S$ to $b_v$ sets $E_1, \ldots, E_{b_v}$, where each $E_i$ represents the set of adjacent edges of vertex $v_i$ in $G'$, in such a way that the following holds: 
        \begin{enumerate}[(i)]
            \item for all $1 \leq i \leq b_v$, $|M_G \cap E_i| \leq 1$;
            \item for all $1 \leq i \leq b_v$, for all $u \in V$, $|S_u \cap E_i| \leq 1$;
            \item the difference of the weights of these sets is upper-bounded by $2W$. 
        \end{enumerate}
            
        We present an algorithm (Algorithm~\ref{algo:distribution-edges}) to 
        obtain the desired distribution. Roughly speaking, we consider all nodes $u \neq v$ one by one. For each $u$, we use a greedy strategy to  distribute the edges of $S_u$ into $s_u$ sets of $E_i$s, where the sets $E_i$ are ordered according to increasing weights, with the exception that the $E_i$s that have not received an edge in $M_G$ so far are the only ones able to receive an edge of $M_G$. 
        
        \begin{algorithm}[h]
        \caption{Algorithm to distribute the edges}\label{algo:distribution-edges}
        \begin{algorithmic}[1]
        \State $\forall\, 1 \leq i \leq b_v, E_i \gets \emptyset$
        \State $\forall\, 1 \leq i \leq b_v, \textsc{UsedForMatching}_i \gets \textsc{False}$
        \For{$u \in V$}
            \State let $i_1,\ldots,i_{m_u}$ be the indices of the $m_u$ sets $E_i$ not $\textsc{UsedForMatching}$ and having the smallest weights, ordered such that $w(E_{i_1}) \leq \ldots \leq w(E_{i_{m_u}})$
            \For{$1 \leq j \leq m_u$}
                \State add the edge $e_{u,j}$ to $E_{i_j}$
                \State $\textsc{UsedForMatching}_{i_j} \gets \textsc{True}$
            \EndFor
            \State let $i_{m_u+1},\ldots,i_{s_u}$ be the indices of the $s_u - m_u$ sets $E_i$ having the smallest weights (excluding the indices $i_1,\ldots,i_{m_u}$), ordered such that $w(E_{i_{m_u + 1}}) \leq \ldots \leq w(E_{i_{s_u}})$
            \For{$m_u < j \leq s_u$}
                \State add the edge $e_{u,j}$ to $E_{i_j}$
            \EndFor
        \EndFor
        \end{algorithmic}
        \end{algorithm}
        
        By construction, it is clear that the two first properties are satisfied.
        We prove the third property by establishing the following invariants:
        
        \begin{itemize}
            \item For any two sets $E_i$ and $E_j$ both with $\textsc{UsedForMatching}$ being \textsc{false}, $|w(E_j)-w(E_i)|\leq W$.  
            \item Given a set $E_i$ with the variable $\textsc{UsedForMatching}$ being \textsc{false} and another set $E_j$ with the variable $\textsc{UsedForMatching}$ being \textsc{true}, $|w(E_j)-w(E_i)|\leq W$. 
            \item For any two sets $E_i$ and $E_j$ both with $\textsc{UsedForMatching}$ being \textsc{true}, $|w(E_j)-w(E_i)|\leq 2W$.  
        \end{itemize}
        
        We can prove it by induction on the number of vertices $u \neq v$ processed so far. To facilitate the proof, we can imagine that we add into $b_v-s_u$ edges of weight 0 into $S_u$ so that these edges are added into the sets $E_{i_{s_u+1}}, \cdots E_{i_{b_v}}$ immediately before the end of the loop. In other words, all sets $E_i$ receive an edge each. Notice that by doing this, we still ensure that the edges of $S_u$ are ordered in their non-increasing weights while the sets $E_{i_1}, \cdots, E_{i_{m_u}}$ and the sets $E_{i_{m_u}+1}, \cdots, E_{i_{b_v}}$ are ordered by non-decreasing weights.
        
        Consider two sets $E_i$ and $E_j$, then by induction hypothesis:
        \begin{itemize}
            \item If $\textsc{UsedForMatching}_i = \textsc{UsedForMatching}_j = \textsc{False}$ before $u \in V$ is processed, then $|w(E_i) - w(E_j)| \leq W$. Suppose that during this step, an edge $e_i$ is added to $E_i$ and an edge $e_j$ is added to $E_j$. If $w(E_i) = w(E_j)$ then we are done, has the difference cannot increase by more than $W$ (this happens if one set receives $W$ and the other receives $0$). Now suppose for instance that $w(E_i) > w(E_j)$. Then, as $E_i$ and $E_j$ have both not received edges from $M_G$ before $u$ is processed, it means that we must have $w(e_i) \leq w(e_j)$, hence we also get that $|w(E_i \cup \{e_i\}) - w(E_j \cup \{e_j\})| \leq W$.
            \item Similarly, if $\textsc{UsedForMatching}_i = \textsc{UsedForMatching}_j = \textsc{True}$ before $u \in V$ is processed, then $|w(E_i) - w(E_j)| \leq 2W$. Suppose that during this step, an edge $e_i$ is added to $E_i$ and an edge $e_j$ is added to $E_j$. If $w(E_i) = w(E_j)$ then we are done, as the difference cannot increase by more than $W$ (this happens if one set receives $W$ and the other receives $0$). Now suppose for instance that $w(E_i) > w(E_j)$. Then, as $E_i$ and $E_j$ have both received edges from $M_G$ before $u$ is processed, it means that we must have $w(e_i) \leq w(e_j)$, hence we also get that $|w(E_i \cup \{e_i\}) - w(E_j \cup \{e_j\})| \leq 2W$.
            \item Finally, if $\textsc{UsedForMatching}_i = \textsc{True}$ and $\textsc{UsedForMatching}_j = \textsc{False}$ before $u$ is processed, then $|w(E_i) - w(E_j)| \leq W$. If processing $u$ the variables $\textsc{UsedForMatching}$ have not changed, then it means that the edges for $E_i$ and $E_j$ have been added during the phase of Lines~8-10 of the algorithm. As a result, the argument for the previous cases also works here. Otherwise, it means that $\textsc{UsedForMatching}_j$ is set to $\textsc{True}$ as well during the process. As a result the edge $e_j$ for $E_j$ has been added during the phase of Lines~4-7 while the edge $e_i$ for $E_i$ has been added during the phase of Lines~8-10. The worst case happens when $w(E_j) - w(E_i) = W$, $w(e_j) = W$, and $w(e_i) = 0$, which leads to $|w(E_j \cup \{e_j\}) - w(E_i \cup \{e_i\})| = 2W$.
        \end{itemize}
        
        To conclude, the maximum weight difference between two sets $E_i$ and $E_j$ in the end is at most $2W$. Moreover, the total sum of the weights is within the interval $[\wdeg_H(v), \wdeg_H(v) + b_v \cdot W]$ (because at most $b_v$ edges for $M_G$ have been added), therefore the weights of all these sets are within the interval $[\frac{\wdeg_H(v)}{b_v} - 2W, \frac{\wdeg_H(v)}{b_v} + 3W]$.
    \end{proof}

\section{Tightness of the Bound for Weighted EDCSes}

    \label{app:tight-bound-edcs}
    
    \begin{proposition}
        Let $W$ be an integer. Let $\beta \geq \beta^- + 2$ be integer parameters. Suppose that $\beta^-$ is divisible by $2W$. Then, there exists a graph $G$ and a $(\beta,\beta^-)$-$w$-EDCS $H$ of $G$ such that the heaviest weighted matching in $H$ is at best a $2 + \frac{1}{\beta^-} - \frac{1}{2W}$ approximation of the largest weighted matching in $G$.
    \end{proposition}
    
    \begin{proof}
        Consider the example depicted in Figure~\ref{fig:example-approx-ratio}. Let $\beta^- = 2kW$. The dashed lines (not part of $H$) satisfy Property~(ii) of Definition~\ref{def:w-edcs}. 
        We set $l$ such that $\beta = l + k + 1$ so that the solid lines representing the edges $(c_i,b_i)$ and $(d_i,e_i)$ satisfy Property~(i) of Definition~\ref{def:w-edcs}. Here $l$ can be chosen to be large enough so that $\beta = \beta^- + 2$. 
        
        By this construction, we have $w(M_H) = 2kW$ and \[\frac{w(M_G)}{w(M_H)} = \frac{2kW + l}{2kW} = 1 + \frac{\beta - k - 1}{2kW} = 1 + \frac{\beta - 1}{\beta^-} - \frac{1}{2W} \geq 2 + \frac{1}{\beta^-} - \frac{1}{2W}.\]

        \begin{figure}[h]
            \centering
            \begin{tikzpicture}
                \node[draw, circle] (A1) at (0,-2) {$a_1$};
                \node[draw, circle] (A2) at (1,-2) {$a_2$};
                \node[circle] (Ap) at (2,-2) {$\cdots$};
                \node[draw, circle] (Ak) at (3,-2) {$a_k$};
                
                \node[draw, circle] (B1) at (0,2) {$b_1$};
                \node[draw, circle] (B2) at (1,2) {$b_2$};
                \node[circle] (Bp) at (2,2) {$\cdots$};
                \node[draw, circle] (Bk) at (3,2) {$b_k$};
                
                \node[draw, circle] (C1) at (5,-2) {$c_1$};
                \node[draw, circle] (C2) at (6,-2) {$c_2$};
                \node[circle] (Cp) at (7,-2) {$\cdots$};
                \node[draw, circle] (Cl) at (8,-2) {$c_l$};
                
                \node[draw, circle] (D1) at (5,2) {$d_1$};
                \node[draw, circle] (D2) at (6,2) {$d_2$};
                \node[circle] (Dp) at (7,2) {$\cdots$};
                \node[draw, circle] (Dl) at (8,2) {$d_l$};
                
                \node[draw, circle] (E1) at (10,-2) {$e_1$};
                \node[draw, circle] (E2) at (11,-2) {$e_2$};
                \node[circle] (Ep) at (12,-2) {$\cdots$};
                \node[draw, circle] (Ek) at (13,-2) {$e_k$};
                
                \node[draw, circle] (F1) at (10,2) {$f_1$};
                \node[draw, circle] (F2) at (11,2) {$f_2$};
                \node[circle] (Fp) at (12,2) {$\cdots$};
                \node[draw, circle] (Fk) at (13,2) {$f_k$};
                
                \draw [-,red] (A1) -- (B1);
                \draw [-,red] (A2) -- (B2);
                \draw [-,red] (Ak) -- (Bk);
                
                \draw [-] (C1) -- (B1);
                \draw [-] (C1) -- (B2);
                \draw [-] (C1) -- (Bk);
                
                \draw [-] (C2) -- (B1);
                \draw [-] (C2) -- (B2);
                \draw [-] (C2) -- (Bk);
                
                \draw [-] (Cl) -- (B1);
                \draw [-] (Cl) -- (B2);
                \draw [-] (Cl) -- (Bk);
                
                \draw [-,dashed, red] (C1) -- (D1);
                \draw [-,dashed, red] (C2) -- (D2);
                \draw [-,dashed, red] (Cl) -- (Dl);
                
                \draw [-] (D1) -- (E1);
                \draw [-] (D1) -- (E2);
                \draw [-] (D1) -- (Ek);
                
                \draw [-] (D2) -- (E1);
                \draw [-] (D2) -- (E2);
                \draw [-] (D2) -- (Ek);
                
                \draw [-] (Dl) -- (E1);
                \draw [-] (Dl) -- (E2);
                \draw [-] (Dl) -- (Ek);
                
                \draw [-,red] (E1) -- (F1);
                \draw [-,red] (E2) -- (F2);
                \draw [-,red] (Ek) -- (Fk);

            \end{tikzpicture}
            \caption{The $w$-EDCS $H$ is represented by solid lines. All the solid edges have weight $W$. The dashed edges are those that were not taken in the $w$-EDCS $H$. All the dashed edges have weight $1$, so that we must have $k \cdot W \geq \beta^-$. Edges in red are part of the optimal matching.}
            \label{fig:example-approx-ratio}
        \end{figure}
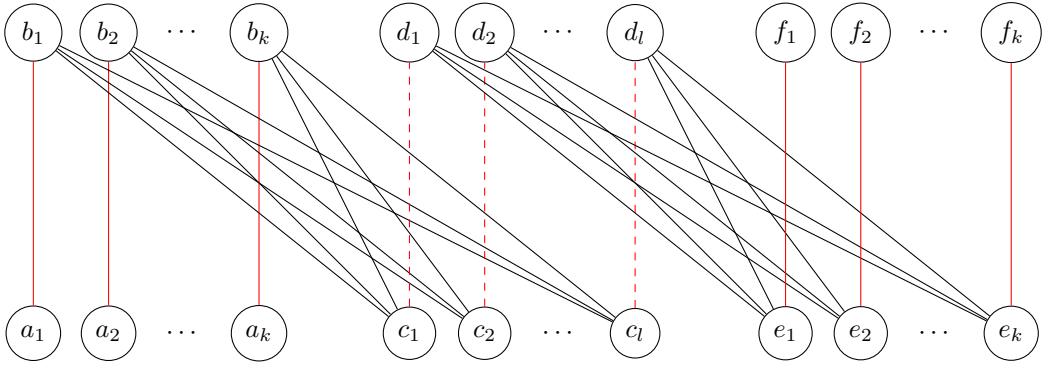
    \end{proof}

\section{Using $W$ Copies Simple EDCSes does not Work}

    \label{app:copies-edcs}
    
    One could think that using $W$ copies of edge-degree constrained subgraphs, one for each weight class, could lead to a structure containing a good approximation of the optimal matching. Here we show that this idea does not work. In fact, consider the simple example where the EDCS lead to a $3/2$ approximation, as depicted in Figure~\ref{fig:example-approx-ratio-2-3}, where $\beta =  2k + 1$ and $\beta^- = 2k$. We represent this situation in a more compacted form, as in Figure~\ref{fig:compact-graph}.
    
    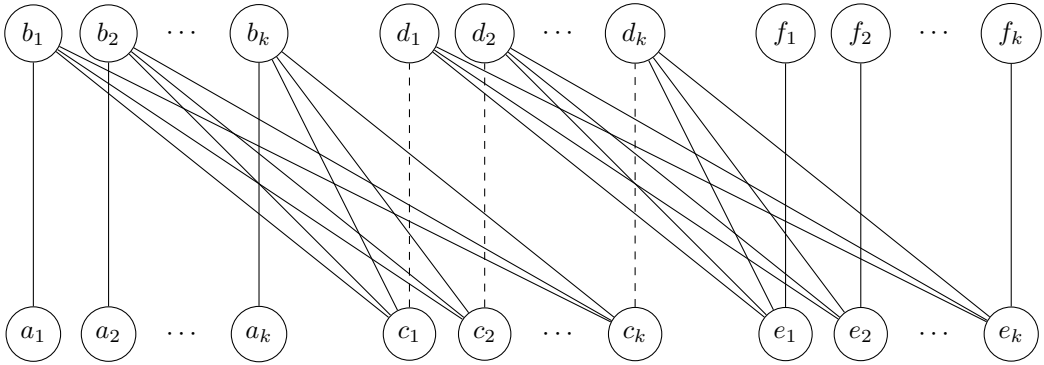
\begin{figure}[h]
        \centering
        \begin{tikzpicture}
            \node[draw, circle] (A1) at (0,-2) {$a_1$};
            \node[draw, circle] (A2) at (1,-2) {$a_2$};
            \node[circle] (Ap) at (2,-2) {$\cdots$};
            \node[draw, circle] (Ak) at (3,-2) {$a_k$};
            
            \node[draw, circle] (B1) at (0,2) {$b_1$};
            \node[draw, circle] (B2) at (1,2) {$b_2$};
            \node[circle] (Bp) at (2,2) {$\cdots$};
            \node[draw, circle] (Bk) at (3,2) {$b_k$};
            
            \node[draw, circle] (C1) at (5,-2) {$c_1$};
            \node[draw, circle] (C2) at (6,-2) {$c_2$};
            \node[circle] (Cp) at (7,-2) {$\cdots$};
            \node[draw, circle] (Cl) at (8,-2) {$c_k$};
            
            \node[draw, circle] (D1) at (5,2) {$d_1$};
            \node[draw, circle] (D2) at (6,2) {$d_2$};
            \node[circle] (Dp) at (7,2) {$\cdots$};
            \node[draw, circle] (Dl) at (8,2) {$d_k$};
            
            \node[draw, circle] (E1) at (10,-2) {$e_1$};
            \node[draw, circle] (E2) at (11,-2) {$e_2$};
            \node[circle] (Ep) at (12,-2) {$\cdots$};
            \node[draw, circle] (Ek) at (13,-2) {$e_k$};
            
            \node[draw, circle] (F1) at (10,2) {$f_1$};
            \node[draw, circle] (F2) at (11,2) {$f_2$};
            \node[circle] (Fp) at (12,2) {$\cdots$};
            \node[draw, circle] (Fk) at (13,2) {$f_k$};
            
            \draw [-] (A1) -- (B1);
            \draw [-] (A2) -- (B2);
            \draw [-] (Ak) -- (Bk);
            
            \draw [-] (C1) -- (B1);
            \draw [-] (C1) -- (B2);
            \draw [-] (C1) -- (Bk);
            
            \draw [-] (C2) -- (B1);
            \draw [-] (C2) -- (B2);
            \draw [-] (C2) -- (Bk);
            
            \draw [-] (Cl) -- (B1);
            \draw [-] (Cl) -- (B2);
            \draw [-] (Cl) -- (Bk);
            
            \draw [-,dashed] (C1) -- (D1);
            \draw [-,dashed] (C2) -- (D2);
            \draw [-,dashed] (Cl) -- (Dl);
            
            \draw [-] (D1) -- (E1);
            \draw [-] (D1) -- (E2);
            \draw [-] (D1) -- (Ek);
            
            \draw [-] (D2) -- (E1);
            \draw [-] (D2) -- (E2);
            \draw [-] (D2) -- (Ek);
            
            \draw [-] (Dl) -- (E1);
            \draw [-] (Dl) -- (E2);
            \draw [-] (Dl) -- (Ek);
            
            \draw [-] (E1) -- (F1);
            \draw [-] (E2) -- (F2);
            \draw [-] (Ek) -- (Fk);

        \end{tikzpicture}
        \caption{All edges are of the same weight. The EDCS $H$ is represented by solid lines. The dashed edges are those that were not taken in the EDCS $H$. Here $\frac{w(M_G)}{w(M_H)}=\frac{3}{2}$}
        \label{fig:example-approx-ratio-2-3}
    \end{figure}
    
    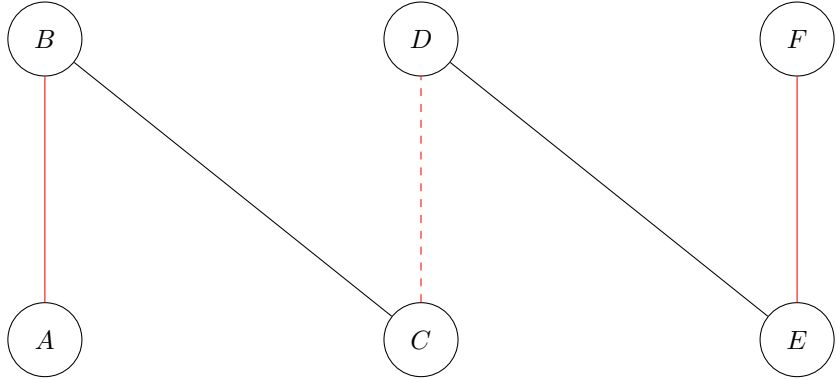
\begin{figure}[h]
        \centering
        \begin{tikzpicture}
            \node[draw, circle, minimum size=28pt] (A1) at (0,-2) {$A$};
            
            \node[draw, circle, minimum size=28pt] (B1) at (0,2) {$B$};
            
            \node[draw, circle, minimum size=28pt] (C1) at (5,-2) {$C$};
            
            \node[draw, circle, minimum size=28pt] (D1) at (5,2) {$D$};
            
            \node[draw, circle, minimum size=28pt] (E1) at (10,-2) {$E$};

            \node[draw, circle, minimum size=28pt] (F1) at (10,2) {$F$};
            
            \draw [-,red] (A1) -- (B1);
            
            \draw [-] (C1) -- (B1);
            
            \draw [-,dashed, red] (C1) -- (D1);
            
            \draw [-,] (D1) -- (E1);

            \draw [-, red] (E1) -- (F1);

        \end{tikzpicture}
        \caption{Each set is of size $k$. The EDCS $H$ is represented by solid lines. The dashed edges are those that were not taken in the EDCS $H$. Red lines represent perfect matchings between sets whereas black lines represent complete bipartite graphs between sets.}
        \label{fig:compact-graph}
    \end{figure}
    
    Using this compacted form, we can build a graph as in Figure~\ref{fig:example-multiple-copies}. For each weight class we have a situation where the approximation ratio $3/2$ is reached, while all the matchings in the EDCS retained are intersecting. As a result, the maximum weight of a matching in the union of EDCS is $2kW$ whereas the actual maximum weight matching is of weight $2kW + \sum_{i = 1}^W k \cdot i = 2kW + k \cdot \frac{W(W+1)}{2}$. Then the approximation ratio for $W \geq 3$ is larger than $2$ and for $W = 2$ the approximation ratio is $\frac{7}{4} = 2 - \frac{1}{2 \cdot 2}$ so for $W = 2$ the approximation ratio is at least the same as the one we can get with a $w$-EDCS.
    
    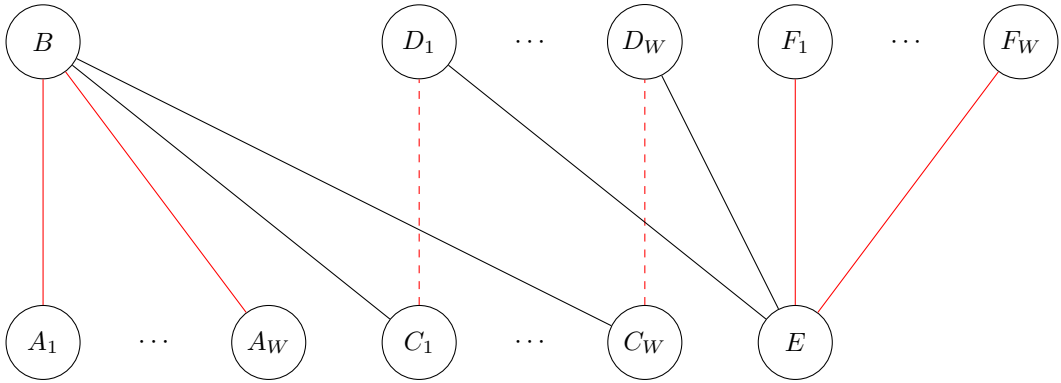
\begin{figure}[h]
        \centering
        \begin{tikzpicture}
            \node[draw, circle, minimum size=28pt] (A1) at (0,-2) {$A_1$};
            \node[circle, minimum size=28pt] (Ap) at (1.5,-2) {$\cdots$};
            \node[draw, circle, minimum size=28pt] (Ak) at (3,-2) {$A_W$};
            
            \node[draw, circle, minimum size=28pt] (B1) at (0,2) {$B$};
            
            \node[draw, circle, minimum size=28pt] (C1) at (5,-2) {$C_1$};
            \node[circle, minimum size=28pt] (Cp) at (6.5,-2) {$\cdots$};
            \node[draw, circle, minimum size=28pt] (Cl) at (8,-2) {$C_W$};
            
            \node[draw, circle, minimum size=28pt] (D1) at (5,2) {$D_1$};
            \node[circle, minimum size=28pt] (Dp) at (6.5,2) {$\cdots$};
            \node[draw, circle, minimum size=28pt] (Dl) at (8,2) {$D_W$};
            
            \node[draw, circle, minimum size=28pt] (E1) at (10,-2) {$E$};
            
            \node[draw, circle, minimum size=28pt] (F1) at (10,2) {$F_1$};
            \node[circle, minimum size=28pt] (Fp) at (11.5,2) {$\cdots$};
            \node[draw, circle, minimum size=28pt] (Fk) at (13,2) {$F_W$};
            
            \draw [-,red] (A1) -- (B1);
            \draw [-,red] (Ak) -- (B1);
            
            \draw [-] (C1) -- (B1);
            \draw [-] (Cl) -- (B1);
            
            \draw [-,dashed,red] (C1) -- (D1);
            \draw [-,dashed,red] (Cl) -- (Dl);
            
            \draw [-] (D1) -- (E1);
            \draw [-] (Dl) -- (E1);
            
            \draw [-,red] (E1) -- (F1);
            \draw [-,red] (E1) -- (Fk);

        \end{tikzpicture}
        \caption{The edge-degree constrained subgraphs for each weight class are represented by solid lines. The dashed edges are those that were not taken in the EDCS. Red lines represent perfect matchings between sets while black lines represent complete bipartite graphs between sets. The edges incident to $A_i$, $C_i$, $D_i$, or $F_i$ are all of weight $i$.}
        \label{fig:example-multiple-copies}
    \end{figure}
    
\section{Streams Containing Irrelevant Edges}

    \label{app:relevant-edges}
    
    Here we no longer consider that the number of edges between two vertices $u$ and $v$ is bounded by $\min(b_u,b_v)$. We therefore introduce the notion of \emph{relevant edges} in a multi-graph for a given set of constraints $\{b_v\}_{v \in V}$.
    \begin{definition}
        In a multi-graph $G$, the multi-set of relevant edges between two given vertices $u$ and $v$ is the multi-set of the $\min(b_u,b_v)$ edges in $G$ between $u$ and $v$ having the largest weights. The relevant subgraph $\overline{G}$ of $G$ is the subgraph of $G$ made of these relevant edges.
    \end{definition}
    Clearly, the only edges relevant to build a $b$-matching are these so-called ``relevant edges'', so we can always assume that $M_G = M_{\overline{G}}$. Moreover, the notion of $w$-$b$-EDCS is relevant with respect to that relevant subgraph.
    
    We can then rephrase Theorem~\ref{thm:intro-general-w-b-edcs} using this notion of relevant subgraph:
    \begin{theorem} \label{thm:intro-general-w-b-edcs-relevant}
        Let $0 < \varepsilon < 1/2$ and let $W$ be an integer parameter. Set $\lambda = \frac{\varepsilon}{100W}$. Let $\beta \geq \beta^- + 2$ be integers such that $\frac{\beta + 6W}{\log(\beta + 6W)} \geq 2 W^2\lambda^{-2}$ and $\beta^- - 6W \geq (1 - \lambda) \cdot (\beta + 6W)$. Then any $(\beta, \beta^-)$-$w$-$b$-EDCS $H$ of the relevant subgraph $\overline{G}$ of an edge-weighted multi-graph $G$ with integer edge weights bounded by $W$ contains a $b$-matching $M_H$ such that $\left(2 - \frac{1}{2W} + \varepsilon\right)\cdot w(M_H) \geq w(M_G)$.
    \end{theorem}

    We can also adapt the arguments of Section~\ref{sec:random-order} to handle the more general case of a stream containing relevant and irrelevant edges, introducing some small variations in our previous argumentation.
    
    \begin{definition}
        We say that a graph $H$ has $b$-bounded weighted edge-degree $\beta$ if for every edge $(u,v, w_{uv}) \in H$, $\frac{\wdeg_H(u)}{b_u} + \frac{\wdeg_H(v)}{b_v} \leq \beta \cdot w_{uv}$ and $|\delta_H(u) \cap \delta_H(v)| \leq \min(b_u,b_v)$. 
    \end{definition}
    
    \begin{definition} 
        Let $G$ be any edge-weighted multi-graph, and let $H$ be a subgraph of $G$ with $b$-bounded weighted edge-degree $\beta$. For any parameter $\beta^-$, we say that an edge $(u,v, w_{uv}) \in G \backslash H$ is $(H,\beta,\beta^-)$-underfull if any of the following conditions is satisfied:
        \begin{enumerate}[(i)]
            \item $|\delta_H(u) \cap \delta_H(v)| < \min(b_u,b_v)$ and $\frac{\wdeg_H(u)}{b_u} + \frac{\wdeg_H(v)}{b_v} < \beta^- \cdot w_{uv}$
            \item $|\delta_H(u) \cap \delta_H(v)| = \min(b_u,b_v)$ and $\min_{e'\in \delta_H(u) \cap \delta_H(v)} w(e') < w_{uv}$
        \end{enumerate}
    \end{definition}
    
    We now show that one can always construct a large matching from the combination of these two parts as well.
    
    \begin{lemma}
        \label{lem:underfull-matching-relevant}
        Let $0 < \varepsilon < 1/2$ be any parameter and $W$ be an integer parameter. Set $\lambda = \frac{\varepsilon}{100W}$. For $\beta \geq \beta^- + 2$ integers such that $\frac{\beta + 8W}{\log(\beta + 8W)} \geq 2 W^2\lambda^{-2}$ and $\beta^- - 6W \geq (1 - \lambda) \cdot (\beta + 8W)$, for any edge-weighted multi-graph $G$ with integer weights in ${1,\dots,W}$, and for any subgraph $H$ with bounded weighted edge-degree $\beta$,  if $X$ contains all edges in $G \backslash H$ that are $(H,\beta,\beta^-)$-underfull, then $(2 - \frac{1}{2W} + \varepsilon) \cdot w(M_{H \cup X}) \geq w(M_G)$.
    \end{lemma}
    
    \begin{proof}
        Let $M_G$ be the maximum $b$-matching in $G$, let $M_G^H = M_G \cap H$ and $M_G^{G \backslash H} = M_G \cap (G \backslash H)$. Let $X^M = X \cap M_G^{G \backslash H}$. We can observe that $w(M_G) = w(M_{H \cup M_G^{G \backslash H}})$.
        
        Now we show that $\overline{H \cup X^M}$ is a $(\beta + 2W, \beta^-)$-$w$-$b$-EDCS of $\overline{H \cup M_G^{G \backslash H}}$. A first observation is that for all $v \in V$, we have the inequalities $\wdeg_H(v) \leq \wdeg_{\overline{H \cup X^M}}(v) \leq \wdeg_{H \cup X^M}(v) \leq \wdeg_H(v) + b_v \cdot W$. We start with Property~(ii). By construction, $X^M$ contains all edges $(u, v, w_{uv})$ in $M_G^{G \backslash H}$ such that $\frac{\wdeg_H(u)}{b_u} + \frac{\wdeg_H(v)}{b_v} < \beta^- \cdot w_{uv}$. Therefore, the remaining edges $(u,v, w_{uv}) \in (H \cup M_G^{G \backslash H}) \backslash (H \cup X^M) = M_G^{G \backslash H} \backslash X^M$ satisfy the inequality $\frac{\wdeg_H(u)}{b_u} + \frac{\wdeg_H(v)}{b_v} \geq \beta^- \cdot w_{uv}$. For Property~(i), for $(u,v,w_{uv}) \in H$, we have $\frac{\wdeg_{H \cup X^M}(u)}{b_u} + \frac{\wdeg_{H \cup X^M}(v)}{b_v} \leq (\beta + 2W)\cdot w_{uv}$ (unless that edge is not relevant in $H \cup M_G^{G \backslash H}$). And for $(u, v, w_{uv}) \in X^M$, $\frac{\wdeg_{H \cup X^M}(u)}{b_u} + \frac{\wdeg_{H \cup X^M}(v)}{b_v} < (\beta^- + 2W)\cdot w_{uv} < (\beta + 2W)\cdot w_{uv}$, so Property~(i) is also satisfied.
        
        As a result, Theorem~\ref{thm:intro-general-w-b-edcs} can be applied in this case and we get that
        $\left(2 - \frac{1}{2W} + \varepsilon\right) \cdot w\left(M_{H \cup X^M}\right) \geq w\left(M_{H \cup M_G^{G \backslash H}}\right) = w(M_G)$, thus concluding the proof.
    \end{proof}
   
    \begin{algorithm}[h]
	\caption{Main algorithm computing a weighted $b$-matching for a random-order stream}\label{algo:w-b-matching-relevant}
	\begin{algorithmic}[1]
	\State $H \gets \emptyset$
	\State $\forall\,0 \leq i \leq \log_2 m,\, \alpha_i \gets \left\lfloor\frac{\varepsilon\cdot m }{\log_2(m) \cdot (2^{i+2}\beta^2W^2 + 1)}\right\rfloor$
	\For{$i = 0 \dots \log_2 m$}
	    \State $\textsc{ProcessStopped} \gets \textsc{False}$
	    \For{$2^{i+2}\beta^2W^2 + 1$ iterations}
	        \State $\textsc{FoundUnderfull} \gets \textsc{False}$
	        \For{$\alpha_i$ iterations}
	            \State let $(u,v, w_{uv})$ be the next edge in the stream
	            \If{$\frac{\wdeg_H(u)}{b_u} + \frac{\wdeg_H(v)}{b_v} < \beta^- \cdot w_{uv}$}
	                \If{$|\delta_H(u) \cap \delta_H(v)| = \min(b_u, b_v)$}
	                    \If{$w_{uv} \leq \min_{e' \in \delta_H(u) \cap \delta_H(v)}w(e')$}
	                        \State \textbf{continue} \Comment{irrelevant edge, we ignore it}
	                    \EndIf
	                    \State remove the smallest edge linking $u$ and $v$
	                \EndIf
	                \State add edge $(u,v, w_{uv})$ to $H$
	                \State $\textsc{FoundUnderfull} \gets \textsc{True}$
                    \While{there exists $(u',v', w_{u'v'}) \in H : \frac{\wdeg_H(u')}{b_{u'}} + \frac{\wdeg_H(v')}{b_{v'}} > \beta \cdot w_{u'v'}$}
                        \State remove $(u', v', w_{u'v'})$ from $H$
                    \EndWhile
	            \EndIf
	        \EndFor
	        \If{$\textsc{FoundUnderfull} = \textsc{False}$}
	            \State $\textsc{ProcessStopped} \gets \textsc{True}$
	            \State \textbf{break} from the loop
	        \EndIf
	    \EndFor
	    \If{$\textsc{ProcessStopped} = \textsc{True}$}
	        \State \textbf{break} from the loop
	    \EndIf
	\EndFor
	\State $X \gets \emptyset$
	\For{each $(u,v,w_{uv})$ remaining edge in the stream}
	    \If{$\frac{\wdeg_H(u)}{b_u} + \frac{\wdeg_H(v)}{b_v} < \beta^- \cdot w_{uv}$}
	        \If{$|\delta_H(u) \cap \delta_H(v)| = \min(b_u, b_v)$ and $w_{uv} \leq \min_{e' \in \delta_H(u) \cap \delta_H(v)}w(e')$}
	                \State \textbf{continue} \Comment{irrelevant edge, we ignore it}
	            \EndIf
	        \State add edge $(u,v, w_{uv})$ to $X$
	    \EndIf
	\EndFor
	\State \Return a maximum weight $b$-matching in $H \cup X$
	\end{algorithmic}
	\end{algorithm}
    
    The algorithm, formally described in Algorithm~\ref{algo:w-b-matching-relevant}, is very similar to the one in Section~\ref{sec:random-order}. However, when there are too many edges between two vertices $u$ and $v$, then another edge cannot be added, unless we remove another one strictly smaller before that (see Lines~10-13). This ``replacement'' operation increases at least by $1$ the value of the potential function used in the construction of a $w$-$b$-EDCS and therefore this operation can be treated like the insertion/deletion operations when bounding the number of such operations. In the later phase of the stream, we also use the new defition of underfull edges (see Lines~26-27).
	
	As we will see later, Algorithm~\ref{algo:w-b-matching-relevant} works when $M_G$ is neither too small nor too big, therefore we have show that these border cases can be handled anyway.
	
	\begin{lemma}
	    \label{lem:bound-graph-size-relevant}
	    Let $\overline{G}$ by the set of relevant edges in $G$. Then, we have $|\overline{G}| \leq  2n \cdot |M_G|$.
	\end{lemma}
	
	\begin{proof}
	    Similar to the proof of Lemma~\ref{lem:bound-graph-size}.
	\end{proof}
	
	\begin{claim} \label{claim:small-output-relevant}
	    We can assume that $w(M_G) \geq \frac{3W^2}{2\varepsilon^2}\log(m)$.
	\end{claim}
	
	\begin{proof}
	    In fact, if $w(M_G) < \frac{3W^2}{2\varepsilon^2}\log(m)$, then $|M_G| < \frac{3W^2}{2\varepsilon^2}\log(m)$ and by Lemma~\ref{lem:bound-graph-size} a semi-streaming algorithm maintaining a \emph{relevant} graph (as defined in Lemma~\ref{lem:bound-graph-size}) would only consume $O(n \cdot \frac{3W^2}{2\varepsilon^2} \cdot \log(m))$ memory. Therefore, such an auxiliary algorithm can be run in parallel of our main algorithm and if the number of edge stored by that auxiliary algorithm reaches $2n \cdot \frac{3W^2}{2\varepsilon^2} \cdot \log(m)$ then we stop that auxiliary algorithm (as it means that $|M_G| \geq \frac{3W^2}{2\varepsilon^2}\log(m)$). Otherwise, the whole graph can be stored so we can compute an exact solution. In both cases we only used $O(n \cdot poly(\log(m),W,1/\varepsilon))$ memory for that auxiliary algorithm.
	\end{proof}
	
	\begin{claim} \label{claim:late-part-relevant}
	    Assuming Claim~\ref{claim:small-output-relevant}, with probability $1 - m^{-3}$ the late part of the stream $E^{late}$ contains at least a $(1 - 2 \varepsilon)$ fraction of an optimal $b$-matching.
	\end{claim}
	
	\begin{proof}
	    Similar to the proof of Claim~\ref{claim:late-part}.
	\end{proof}
	
	\begin{claim} \label{claim:big-output-relevant}
	    We can assume that $\frac{\varepsilon\cdot m}{\log_2(m) \cdot (2^{i_0+2}\beta^2W^2 + 1)} \geq 1$.
	\end{claim}
	
	\begin{proof}
	    Similar to the proof of Claim~\ref{claim:big-output}.
	\end{proof}
	
	Then we can move to our main algorithm.
	
	\begin{lemma} \label{lem:early-properties-relevant}
	    The first phase of Algorithm~\ref{algo:w-b-matching} uses $O(\beta\cdot |M_G|)$ memory and constructs a subgraph $H$ of $G$, satisfying the following properties:
	    \begin{enumerate}
	        \item The first phase terminates within the first $\varepsilon m$ edges of the stream, \emph{i.e.} processing an edge $e_i$ with $i \leq \varepsilon m$.
	        \item When the first phase terminates after processing some edge $e_i$, the subgraph $H \subseteq G$ constructed during this phase satisfies these properties:
	            \begin{enumerate}
	                \item $H$ has $b$-bounded weighted edge degree $\beta$, and hence contains at most $O(\beta \cdot |M_G|)$ edges.
	                \item With probability at least $1 - m^{-3}$, the total number of $(H, \beta, \beta^-)$-underfull edges in the remaining part of the stream is at most $\gamma = O(|M_G|\cdot (\log(m))^2 \cdot \beta^2 W^2 \cdot 1/\varepsilon)$.
	            \end{enumerate}
	    \end{enumerate}
	\end{lemma}
	
	\begin{proof}
	    First, in each interval of size $\alpha_i$ processed until the first phase terminates (except the last interval) there is at least one insertion/deletion operation that is performed (as described in the proof of Proposition~\ref{prop:exist-w-b-edcs}) or a ``replacement'' operation (when an irrelevant edge is replaced by a larger one, see Lines~10-13; note that this operation also increases the potential function in the proof of Proposition~\ref{prop:exist-w-b-edcs} by at least $1$) and therefore the total number of such processed intervals is bounded by $4 \beta^2W^2\cdot |M_G| + 1$. As a result, the first phase ends for some $i \leq i_0 = \lceil \log_2 |M_G| \rceil$, and the total number of edge processed in the first phase is therefore bounded by $\varepsilon m \cdot \frac{i_0}{\log_2(m)} \leq \varepsilon m$.
	    
	    As the subgraph $H$ built always keeps a $b$-bounded weighted edge-degree $\beta$, Proposition~\ref{prop:w-b-edcs-bounded-size} also applies and therefore the construction of $H$ uses $O(\beta \cdot |M_G|)$ memory.

        For the last property, the argument is very similar to the one used in the proof of Lemma~\ref{lem:early-properties} and in~\cite{bernstein:LIPIcs:2020:12419}.
	\end{proof}
	
	Now we can obtain the desired result.
	
	\begin{theorem}
	    \label{thm:streaming-approx-relevant}
	    Let $\varepsilon > 0$ and $W$ be an integer parameter. Using Algorithm~\ref{algo:w-b-matching-relevant}, with probability $1 - 2m^{-3}$, one can extract from a randomly-ordered stream of edges a weighted $b$-matching with an approximation ratio of $2 - \frac{1}{2W} + \varepsilon$, using $O(\max(|M_G|,n) \cdot poly(\log(m), W, 1/\varepsilon))$ memory.
	\end{theorem}
	
	\begin{proof}
	    Applying Lemma~\ref{lem:underfull-matching-relevant} to the graph $H \cup G^{late}$ we can get, if we have chosen the right values $\beta$ and $\beta^-$ (which are $O(poly(W, 1/\varepsilon))$), $H \cup X$ contains a $(1 - 2\varepsilon)^{-1}\cdot (2 - \frac{1}{2W} + \varepsilon)$ approximation of the optimal $b$-matching (with probability at least $1 - m^{-3}$, see Claim~\ref{claim:late-part-relevant}), and with  a memory consumption of $O(|M_G| \cdot poly(\log(m), W, 1/\varepsilon))$ (with probability at least $1 - m^{-3}$, see Lemma~\ref{lem:early-properties-relevant}), with probability at least $1 - 2m^{-3}$ (by union bound). This concludes the proof.
	\end{proof}

\end{document}